\documentclass[11pt]{article}
\usepackage{mysty}
\usepackage{tikz}
\usepackage[maxbibnames=6,minbibnames=5]{biblatex}
\usepackage{appendix}
\usepackage{algorithm}
\usepackage{float}
\usepackage{multirow}
\usepackage{booktabs}
\usepackage{centernot}

\newcommand{\rt}{\operatorname{root}}

\newcommand{\Qbit}{\mathbf{Qbit}}
\newcommand{\Qpar}{\mathbf{Qpar}}
\newcommand{\Act}{\mathbf{Act}}
\newcommand{\Proc}{\mathbf{Proc}}
\newcommand{\Sys}{\mathbf{Sys}}
\newcommand{\QO}{\operatorname{\mathcal{QO}}}

\newlist{defenum}{enumerate}{1}
\setlist[defenum]{label=(\alph*), ref=\thedefinition~(\alph*)}
\crefalias{defenumi}{definition}

\newlist{postenum}{enumerate}{1}
\setlist[postenum]{label=(\alph*), ref=\thedefinition~(\alph*)}
\crefalias{postenumi}{postulate}

\newcommand{\footremember}[2]{%
   \footnote{#2}
    \newcounter{#1}
    \setcounter{#1}{\value{footnote}}%
}

\title{Atomicity in Distributed Quantum Computing}

\author{Zhicheng Zhang\footremember{a}{Centre for Quantum Software and Information, University of Technology Sydney, Sydney, Australia.
	Email: \texttt{iszczhang@gmail.com}.}
\and Mingsheng Ying\footremember{b}{Centre for Quantum Software and Information, University of Technology Sydney, Sydney, Australia.
	Email: \texttt{Mingsheng.Ying@uts.edu.au}.}
}

\begin{document}

\maketitle

\begin{abstract}
	Atomicity is a ubiquitous assumption in distributed computing,
	under which actions are indivisible and appear sequential.
	In classical computing, this assumption has several theoretical and practical guarantees.
	In quantum computing, although atomicity is still commonly assumed,
	it has not been seriously studied,
	and a rigorous basis for it is missing.
	Classical results on atomicity do not directly carry over to distributed quantum computing,
	due to new challenges caused by quantum entanglement and the measurement problem from the underlying quantum mechanics.

	In this paper, we initiate the study of atomicity in distributed quantum computing.
	A formal model of (non-atomic) distributed quantum system is established.
	Based on the Dijkstra-Lamport condition, 
	the system dynamics and observable dynamics of a distributed quantum system are defined,
	which correspond to the quantum state of and classically observable events in the system, respectively.
	Within this framework,
	we prove that local actions can be regarded as if they were atomic,
	up to the observable dynamics of the system.
\end{abstract}

\newpage

{
	\hypersetup{linkcolor=black}
	\tableofcontents
}

\newpage

\section{Introduction}
\label{sec:introduction}

Distributed system is a collection of processes taking actions concurrently
and expected to cooperate harmoniously for some computational goal.
Most models of distributed systems 
build upon the assumption of \textit{atomic actions}.
Atomic actions are indivisible and mutually exclusive in time:
for any two atomic actions $a$ and $b$ (probably from different processes),
either $a$ precedes $b$, or $b$ precedes $a$.
In comparison, for two non-atomic actions,
it is possible that they are concurrent and have no temporal order at all~\cite{Lamport86a}.
Assuming the atomicity greatly reduces the \textit{non-determinism} introduced by concurrency.
As a result, the behaviour of a distributed system
can be simplified as a sequence of atomic actions and easier for reasoning.

\paragraph{Atomicity assumption}
Taking a closer look at the atomicity assumption,
it consists of the following two parts:
\begin{enumerate}
	\item 
		Atomicity of \textbf{\textit{non-local
		actions}} that may overlap in space-time with others 
		(e.g., read/write actions on shared variables):		
		In classical computing, if the overlap 
		is associated with the hardware like a memory location,
		often the hardware can support atomicity
		(e.g., atomic instruction compare-and-swap is provided in the x86 architecture).
		Lower-level atomic actions can be used to construct higher-level ones.
		Well-known examples include implementing the synchronising primitive ``semaphore'' with atomic instructions (e.g., test-and-set).
		This paradigm of layered design is introduced and promoted by the early pioneer Dijkstra~\cite{Dijkstra67,EWD310},
		who also identified a closely related key problem in concurrency, the mutual exclusion problem~\cite{Dijkstra65}.

		Perhaps surprisingly, even without a strong hardware support,
		this part of atomicity can also be guaranteed by pure software.
		Lamport introduced the notions of safe and atomic registers~\cite{Lamport86c,Lamport86d}:
		the former is weak and can be implemented 
		without assuming lower-level atomicity;
		while the latter is strong and actions on which are promised to be atomic.
		It is shown that atomic registers can be constructed from safe registers.
		The underlying model of this result is closely related to a line of studies by
		Lamport~\cite{Lamport74,Lamport79,Lamport86a,Lamport86b} 
		on non-atomic solutions to the mutual exclusion problem.
	\item
		Atomicity of \textbf{\textit{local actions}}\footnote{
			Here, the adjective ``local/non-local'' describes an action.
			In the classical literature, private variables are sometimes also called local variables,
			which is NOT an adopted terminology in this paper.
			Note that a local action does NOT necessarily act on private variables.
		}
		that do not overlap in space-time with others
		(see \Cref{def:local-act}): 
		At first glance, this part of atomicity seems to be trivially granted.
		Intuitively, two local actions 
		have no way to interfere with each other,
		and thus assuming a temporal order between them 
		should not affect the overall effect of the system. 

		But what represents the effect (a.k.a., semantics) of a system?
		Dijkstra is the first to notice the subtleties in this question.
		He realised that in a sequential process,
		the effect of an action is a change of the current state of the process,
		and the notion of state is only meaningful at discrete instants right before and after actions~\cite{EWD198}.
		During an action, the state is not defined.
		For a system of multiple asynchronous processes,
		the discrete instants corresponding to different processes are rarely aligned.
		To describe a state of the system,
		one can only take the \textit{direct product} of the states of all processes (at possibly different instants).
		In other words, 
		at any time each process corresponds to a subspace,
		and its state is a projection of the system state onto this subspace.
		Returning to the atomicity of local actions, 
		two concurrent local actions must correspond to separated subspaces.
		It can be shown that two concurrent actions in separated subspaces
		always change a system state to another determined one,
		no matter how they are concurrently implemented.

		We remark that the above justification particularly relies on the ability 
		to take direct product and projection of classical states.
		As will be seen later, 
		such classical reasoning will face challenges when we consider the atomicity in distributed quantum computing,
		because the states in quantum mechanics have quite different properties.
		The picture will become clearer when we discuss some motivating examples in \Cref{sec:motivating_examples}.
		The major contribution of this paper is also about this part of atomicity.
\end{enumerate}

\paragraph{Quantum computing}
Now let us move from classical to quantum. 
The last few decades have witnessed a rapid development of quantum computing.
Numerous quantum algorithms are proposed and shown advantageous on certain problems compared to their classical counterparts.
Celebrated examples include Shor's quantum algorithm for factoring~\cite{Shor94},
and Grover's quantum algorithm for database search~\cite{Grover96}.
While near-term applications of quantum computing are restricted by the scale of existing quantum hardware,
there are already industrial plans to connect multiple quantum processors 
and ultimately build distributed quantum computers.
For example, in IBM Quantum's development roadmap~\cite{IBMQroadmap},
they plan to realise communication between quantum processors in 2025,
and their target beyond 2030 is to build distributed quantum 
systems using long-range quantum networks.

While being ambitious on our way towards future quantum computing,
fundamental concepts like atomicity are, however, \textit{not seriously examined} in this new context.
In distributed quantum computing,
almost all existing models (e.g., quantum process algebra~\cite{JL04,GN05,FDY11}, LOCAL model~\cite{GKM09}, CONGEST model~\cite{EKNP14},
distributed programming~\cite{YZLF22,FLY22,HSHT21,WZLSXD22})
and algorithms (e.g., leader election~\cite{TKM12,AGM17}, dining philosophers~\cite{AGM17}, diameter computation~\cite{LM18},
all-pairs shortest path~\cite{IL19}) 
either implicitly or explicitly assume underlying atomic actions.
We therefore raise the following question:
\begin{equation*}
	\textit{\textbf{Is the atomicity assumption rigorously guaranteed in distributed quantum systems?}}
\end{equation*}
\begin{enumerate}
	\item 
		For the atomicity of non-local actions, 
		we need no worry.
		In practice, a quantum processor is a quantum device controlled by a classical process,
		and before quantum processors take actions on shared objects 
		(e.g., access a quantum register or a quantum communication channel),
		we can let their classical controlling processes communicate and enforce the mutual exclusion at the classical level.
		As aforementioned, this can be done by either (classical) hardware or software.
	\item
		By contrast, for the atomicity of local actions, 
		we do need to worry.
		The aforementioned classical justification relies on the classical properties of states and actions,
		which should be reexamined in the quantum context.
		It turns out that new challenges caused by quantum entanglement and the measurement problem
		from the underlying quantum mechanics make this problem non-trivial (see motivating examples in \Cref{sec:motivating_examples}).
\end{enumerate}

\paragraph{Challenges from quantum}
Quantum mechanics changes our very notion of state.
A classical state can be recorded by bits.
For example, the state space of a single bit is $\braces*{0,1}$,
and that of $n$ bits is $\braces*{0,1}^n$.
The quantum counterpart of bit is \textit{qubit},
of which the state space is a Hilbert space $\calH=\Co^{2}$.
The state of a qubit can be represented by a complex vector $\alpha\ket{0}+\beta\ket{1}$
with $\abs*{\alpha}^2+\abs*{\beta}^2=1$,
a \textit{quantum superposition} of the basis states $\ket{0}$ and $\ket{1}$.
For $n$ qubits, the state space is the tensor product $\calH^{\otimes n}$.
A state in $\calH^{\otimes n}$ can be represented by a vector $\ket{\psi}=\sum_{x\in \braces*{0,1}^n}\alpha_x\ket{x}$
with $\sum_{x}\abs*{\alpha_x}^2=1$.
Quantum superposition results in the phenomenon of \textit{quantum entanglement}.
A quantum state $\ket{\psi}$ is entangled if it is not a product state of the form $\ket{\psi_1}\otimes\ket{\psi_2}$.
A minimal example of entangled state is an EPR pair $\frac{1}{\sqrt{2}}\parens*{\ket{0}\otimes \ket{0}+\ket{1}\otimes \ket{1}}$.
The challenge posed by entanglement is that 
we can no longer break down the state of a system into a product of those of processes within,
because it can be entangled.
Quantum entanglement is a manifestation of the broader concept of quantum non-locality,
which is further discussed in \Cref{sub:id-challenge}.

Actions in quantum computing are also new.
There are two types of quantum operations that an action can perform:
quantum gates and quantum measurements.
A quantum gate is a unitary transformation of the quantum state.
A quantum (partial) measurement collapses a quantum state to another with certain probability
and yields a classical outcome.
A notorious problem in quantum mechanics is the \textit{measurement problem},
roughly stating that we do not know how the state evolves during a 
quantum measurement.\footnote{
	Careful readers might notice the similarity with Dijkstra's observation that during an action the state is not defined, as discussed above.
}
The challenge posed by the measurement problem
is that even if we dare to dive into the (hardware) implementation details of actions,
we cannot certainly determine the real-time dynamics of the state of a system,
in particular, during measurements.

\paragraph{Importance of rigorous basis}
We would like to stress the importance of why we need a rigorous basis for the atomicity.
Design and reasoning of distributed systems are notoriously error-prone,
largely due to the much non-determinism from concurrency.
Take the mutual exclusion problem for example.
Dijkstra described his solution to it as ``by far the most difficult pieces of program I ever made''~\cite{EWD310}.
As recalled by Lamport, in the early years of the field of concurrency,
some published concurrent algorithms later turned out to be incorrect~\cite{Lamport19}.
Even for Lamport's celebrated bakery algorithm,
its original correctness proof~\cite{Lamport74} 
contained a fallacious assumption and was discovered some years later in \cite{Lamport86a,Lamport86b};
and it was only after~\cite{Lamport90b} that two more unrevealed assumptions 
(in different previous versions of proofs~\cite{Lamport74,Lamport77,Lamport79}) were found.
In~\cite{Lamport90b} he pointed out ``the danger in trying to replace one program with an equivalent one,
if the equivalence has not been proved formally''.
Atomicity in distributed quantum computing is exactly a such assumption of equivalence.

\subsection{Motivating Examples}
\label{sec:motivating_examples}
 
Let us start with a series of motivating examples,
each building upon the previous.
Through these examples,
we hope that the readers will be convinced that
the aforementioned classical justification of the atomicity of local actions
does not directly carry over to the quantum case,
and there are many subtleties underlying the problem in the quantum setting.

For illustration, we temporally restrict ourselves to 
\textit{terminating} and \textit{deterministic} quantum systems \textit{without} non-local actions.
The task is to prove the equivalence between a \textit{synchronous} and an \textit{asynchronous} real-time\footnote{
	In this paper, we assume the existence of global (physical) time,
	which is not equivalent to a global (logical) clock.
	Different processes in a system can use different logical clocks.
	The only excepted consideration by this assumption is the relativistic effect,
	which is left for future works.
}
quantum systems.
The former describes an ideal real-time implementation of a logical quantum system:
actions from different processes are synchronised.
The latter describes a possible actual real-time implementation:
actions from different processes are not synchronised.

Systems with atomic actions are definitely synchronous.
At this point, the equivalence of two systems means that given any input state,
they produce the same output state.
So, the task considered in these examples is weaker than justifying the atomicity of local actions, which is the aim of this paper.

\begin{example}
	\label{eg:mot1}
	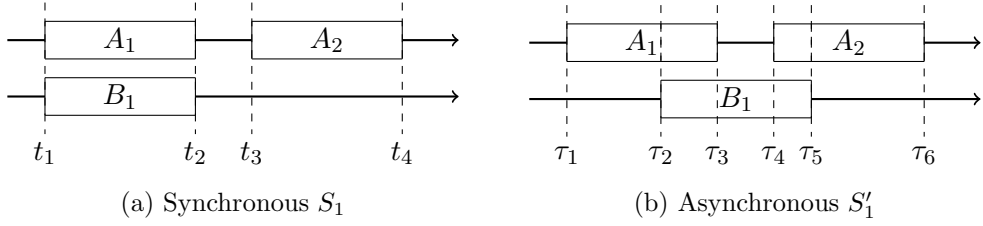
\begin{figure}
		\centering
		\begin{subfigure}[b]{0.4\textwidth}
			\centering
			\begin{tikzpicture}
				\node (0) at (-5, 1.5) {};
				\node (1) at (-5, 1) {};
				\node (2) at (-3, 1.5) {};
				\node (4) at (-3, 1) {};
				\node (9) at (-5, 0.75) {};
				\node (10) at (-5, 0.25) {};
				\node (11) at (-3, 0.75) {};
				\node (12) at (-3, 0.25) {};
				\node (13) at (-5.5, 1.25) {};
				\node (14) at (-5, 1.25) {};
				\node (15) at (-5.5, 0.5) {};
				\node (16) at (-5, 0.5) {};
				\node (17) at (-2.25, 1.5) {};
				\node (18) at (-2.25, 1) {};
				\node (19) at (-0.25, 1.5) {};
				\node (20) at (-0.25, 1) {};
				\node (21) at (-3, 1.25) {};
				\node (22) at (-2.25, 1.25) {};
				\node (23) at (-3, 0.5) {};
				\node (24) at (0.5, 0.5) {};
				\node (25) at (-0.25, 1.25) {};
				\node (26) at (0.5, 1.25) {};
				\node (27) at (-4, 1.25) {$A_1$};
				\node (28) at (-4, 0.5) {$B_1$};
				\node (29) at (-1.25, 1.25) {$A_2$};
				\node (30) at (-5, 1.75) {};
				\node (31) at (-5, 0) {};
				\node (32) at (-5, -0.25) {$t_1$};
				\node (33) at (-3, 1.75) {};
				\node (34) at (-3, 0) {};
				\node (35) at (-3, -0.25) {$t_2$};
				\node (36) at (-2.25, 1.75) {};
				\node (37) at (-2.25, 0) {};
				\node (38) at (-2.25, -0.25) {$t_3$};
				\node (39) at (-0.25, 1.75) {};
				\node (40) at (-0.25, 0) {};
				\node (41) at (-0.25, -0.25) {$t_4$};
				\draw (0.center) to (1.center);
				\draw (0.center) to (2.center);
				\draw (1.center) to (4.center);
				\draw (2.center) to (4.center);
				\draw (9.center) to (10.center);
				\draw (9.center) to (11.center);
				\draw (10.center) to (12.center);
				\draw (11.center) to (12.center);
				\draw [thick] (13.center) to (14.center);
				\draw [thick] (15.center) to (16.center);
				\draw (17.center) to (18.center);
				\draw (17.center) to (19.center);
				\draw (18.center) to (20.center);
				\draw (19.center) to (20.center);
				\draw [thick] (21.center) to (22.center);
				\draw [->,thick] (25.center) to (26.center);
				\draw [->,thick] (23.center) to (24.center);
				\draw [dashed] (30.center) to (31.center);
				\draw [dashed] (33.center) to (34.center);
				\draw [dashed] (36.center) to (37.center);
				\draw [dashed] (39.center) to (40.center);
			\end{tikzpicture}
			\caption{Synchronous $S_1$}
			\label{fig:egmot-1a}
		\end{subfigure}
		\begin{subfigure}[b]{0.4\textwidth}
			\centering
			\begin{tikzpicture}
				\node (0) at (-5, 1.5) {};
				\node (1) at (-5, 1) {};
				\node (2) at (-3, 1.5) {};
				\node (4) at (-3, 1) {};
				\node (9) at (-3.75, 0.75) {};
				\node (10) at (-3.75, 0.25) {};
				\node (11) at (-1.75, 0.75) {};
				\node (12) at (-1.75, 0.25) {};
				\node (13) at (-5.5, 1.25) {};
				\node (14) at (-5, 1.25) {};
				\node (15) at (-5.5, 0.5) {};
				\node (16) at (-3.75, 0.5) {};
				\node (17) at (-2.25, 1.5) {};
				\node (18) at (-2.25, 1) {};
				\node (19) at (-0.25, 1.5) {};
				\node (20) at (-0.25, 1) {};
				\node (21) at (-3, 1.25) {};
				\node (22) at (-2.25, 1.25) {};
				\node (23) at (-1.75, 0.5) {};
				\node (24) at (0.5, 0.5) {};
				\node (25) at (-0.25, 1.25) {};
				\node (26) at (0.5, 1.25) {};
				\node (27) at (-4, 1.25) {$A_1$};
				\node (28) at (-2.75, 0.5) {$B_1$};
				\node (29) at (-1.25, 1.25) {$A_2$};
				\node (30) at (-5, 1.75) {};
				\node (31) at (-5, 0) {};
				\node (32) at (-5, -0.25) {$\tau_1$};
				\node (33) at (-3, 1.75) {};
				\node (34) at (-3, 0) {};
				\node (35) at (-3, -0.25) {$\tau_3$};
				\node (36) at (-2.25, 1.75) {};
				\node (37) at (-2.25, 0) {};
				\node (38) at (-2.25, -0.25) {$\tau_4$};
				\node (39) at (-0.25, 1.75) {};
				\node (40) at (-0.25, 0) {};
				\node (41) at (-0.25, -0.25) {$\tau_6$};
				\node (42) at (-3.75, 1.5) {};
				\node (43) at (-3.75, 1) {};
				\node (44) at (-3.75, 1.25) {};
				\node (45) at (-3.75, 1.75) {};
				\node (46) at (-3.75, 0) {};
				\node (47) at (-3.75, -0.25) {$\tau_2$};
				\node (48) at (-1.75, 1.5) {};
				\node (49) at (-1.75, 1) {};
				\node (50) at (-1.75, 1.25) {};
				\node (51) at (-1.75, 1.75) {};
				\node (52) at (-1.75, 0) {};
				\node (53) at (-1.75, -0.25) {$\tau_5$};
				\draw (0.center) to (1.center);
				\draw (0.center) to (2.center);
				\draw (1.center) to (4.center);
				\draw (2.center) to (4.center);
				\draw (9.center) to (10.center);
				\draw (9.center) to (11.center);
				\draw (10.center) to (12.center);
				\draw (11.center) to (12.center);
				\draw [thick] (13.center) to (14.center);
				\draw [thick] (15.center) to (16.center);
				\draw (17.center) to (18.center);
				\draw (17.center) to (19.center);
				\draw (18.center) to (20.center);
				\draw (19.center) to (20.center);
				\draw [thick] (21.center) to (22.center);
				\draw [->,thick] (25.center) to (26.center);
				\draw [->,thick] (23.center) to (24.center);
				\draw [dashed] (30.center) to (31.center);
				\draw [dashed] (33.center) to (34.center);
				\draw [dashed] (36.center) to (37.center);
				\draw [dashed] (39.center) to (40.center);
				\draw [dashed] (42.center) to (43.center);
				\draw [dashed] (45.center) to (46.center);
				\draw [dashed] (48.center) to (49.center);
				\draw [dashed] (51.center) to (52.center);
			\end{tikzpicture}
			\caption{Asynchronous $S_1'$}
			\label{fig:egmot-1b}
		\end{subfigure}
		\caption{Proving the synchronous system $S_1$ is equivalent to asynchronous $S_1'$.
		Both can be thought of as space-time diagrams. 
		Each wire corresponds to a qubit,
		all together spanning the space.
		Each box corresponds to an action, which on the space axis specifies the qubits it acts on,
		and on the time axis specifies the time interval it spans.
		The arrow of time is from the left to the right.
		}
		\label{fig:egmot-1}
	\end{figure}
	In \Cref{fig:egmot-1a} is a synchronous system $S_1$,
	with two qubits and three actions (which are supposed to be unitary actions).
	We can think of it as a space-time diagram.
	Each wire represents a qubit, all together spanning the space;
	and each box represents an action,
	which on the space axis specifies the qubits it acts on,
	and on the time axis specifies the time interval it spans.
	$S_1$ is synchronous because the time intervals of actions are aligned.
	Let $\ket{\psi\parens*{t}}$ be the state of all qubits in $S_1$ at time $t$.
	We can identify states at the following instants:
	initial state $\ket{\psi(t_1)}$,
	intermediate states $\ket{\psi(t_2)},\ket{\psi\parens*{t_3}}$,
	and output state $\ket{\psi(t_4)}$.
	We can also describe the relations between them:
	$\ket{\psi(t_2)}=\parens*{A_1\otimes B_1}\ket{\psi(t_1)}$,
	$\ket{\psi\parens*{t_3}}=\ket{\psi\parens*{t_2}}$
	and $\ket{\psi(t_4)}=\parens*{A_2\otimes \Id}\ket{\psi\parens*{t_3}}$.\footnote{
		Here, for simplicity, we use the same notation for an action and the operation it performs.
		$\Id$ stands for the identity operator.
	}
	The effect of the system is described by the relation $\ket{\psi\parens*{t_4}}=\parens*{A_2A_1\otimes B_1}\ket{\psi\parens*{t_1}}$.

	In \Cref{fig:egmot-1b} is an asynchronous system $S_1'$,
	expected to be equivalent to $S_1$.\footnote{Readers who are familiar with quantum computing might think this is obvious,
		because $S_1$ and $S_1'$ can be described by the same logical quantum circuit.
		But why do they correspond to the same logical quantum circuit?
		In other words, what convinces us that the actual real-time quantum system faithfully implements a logical one?
		This question is exactly what we need to answer in these examples.
	}
	To prove the equivalence, denoted by $S_1\simeq S_1'$, 
	we need to show they have the same effect.
	$S_1'$ is asynchronous: the time interval of $B_1$ overlaps with those of $A_1$ and $A_2$.
	This is possible when $A_1,A_2$ are from some process $A$,
	and $B_1$ is from some other process $B$.
	Let $\ket{\psi'(\tau)}$ be the state of $S_1'$ at time $\tau$.
	We can still identify states at several instants:
	initial state $\ket{\psi'\parens*{\tau_1}}$,
	output state $\ket{\psi'\parens*{\tau_6}}$,
	and other intermediate states $\ket{\psi'\parens*{\tau_j}}$.
	Now the relations between them involve non-determinism.
	For example, we can only obtain relation like 
	$\ket{\psi'\parens*{\tau_3}}=\parens*{A_1\otimes B_1'}\ket{\psi'\parens*{\tau_1}}$,
	for some unknown ``partial action'' $B_1'$ (depending on the implementation details) of $B_1$.
	In this case, even describing the effect of $S_1'$ using actions within,
	i.e., relating $\ket{\psi'\parens*{\tau_1}}$ and $\ket{\psi'\parens*{\tau_6}}$ using $A_1,A_2$ and $B_1$,
	becomes subtle.

	There is however a way to circumvent the above issue.
	The Hilbert space of two qubits is spanned by the computational basis states $\braces*{\ket{00},\ket{01},\ket{10},\ket{11}}$,
	and by linearity of unitary operators it suffices to verify that the two systems produce the same output states on all these product states.
	It is easy to see the state remains product at all time. 
	Suppose that $\ket{\psi'\parens*{\tau}}=\ket{\psi_1'\parens*{\tau}}\ket{\psi_2'\parens*{\tau}}$,\footnote{
		We use the convention $\ket{\psi}\ket{\phi}=\ket{\psi}\otimes\ket{\phi}.$
	} 
	where $\ket{\psi_1'}$ and $\ket{\psi_2'}$ are states of the first and second qubits, respectively.
	Then, we have
	$\ket{\psi_1'\parens*{\tau_6}}=A_2\ket{\psi_1'\parens*{\tau_4}}=A_2\ket{\psi_1'\parens*{\tau_3}}=A_2A_1\ket{\psi_1'\parens*{\tau_1}}$,
	and $\ket{\psi_2'\parens*{\tau_6}}=\ket{\psi_2'\parens*{\tau_5}}=B_1\ket{\psi_2'\parens*{\tau_2}}=B_1\ket{\psi_2'\parens*{\tau_1}}$.
	Combined with $\ket{\psi\parens*{t_1}}=\ket{\psi'\parens*{\tau_1}}$ 
	and $\ket{\psi\parens*{t_4}}=\ket{\psi'\parens*{\tau_6}}$,
	we can conclude $S_1\simeq S_1'$.


	\begin{remark}
		It is worth stressing again the role played by the \textit{quantum entanglement} here:
		writing $\ket{\psi'\parens*{\tau}}$ 
		as a product $\ket{\psi_1'\parens*{\tau}}\ket{\psi_2'\parens*{\tau}}$
		is impossible if $\ket{\psi\parens*{\tau}}$ is entangled.
		The need to verify the equivalence on the computational basis states
		is because they are product states.
	\end{remark}

	\begin{remark}[Real-time states vs.\ discrete abstract states]
		\label{rmk:con-Dijk-obser}
		What is the connection of the above circumvention with Dijkstra's observation that
		the state cannot be defined during actions?
		Dijkstra's notion of state is abstract and a part of the discrete model.
		The state $\ket{\psi\parens*{t}}$ (and $\ket{\psi'(\tau)}$) we consider here is real-time,
		to which a physical reality can be assigned.
		Abstract states are only meaningful (as good abstractions) when they correspond to physical states.
		Typically, an abstract state is a function of a set of real-time states,
		omitting implementation details unconcerned (hence ``undefined'').
		For example, in $S_1'$,
		we can use real-time states
		$\ket{\psi_1'\parens*{\tau}}$ for $\tau=\tau_1,\tau_3,\tau_4,\tau_6$
		and $\ket{\psi_2'\parens*{\tau}}$ for $\tau=\tau_1,\tau_2,\tau_5,\tau_6$
		to compose abstract states,
		without knowing how $A_1,A_2$ and $B_1$ are implemented.
		This is exactly why the classical reasoning works after the above circumvention.
	\end{remark}
\end{example}

\begin{example}
	\label{eg:mot2}
	\begin{figure}
		\centering
		\begin{subfigure}[b]{0.4\textwidth}
			\centering
			\begin{tikzpicture}
				\node (0) at (-4.5, 1.5) {};
				\node (1) at (-4.5, 1) {};
				\node (2) at (-3, 1.5) {};
				\node (4) at (-3, 1) {};
				\node (9) at (-4.5, 0.75) {};
				\node (10) at (-4.5, 0.25) {};
				\node (11) at (-3, 0.75) {};
				\node (12) at (-3, 0.25) {};
				\node (13) at (-5, 1.25) {};
				\node (14) at (-4.5, 1.25) {};
				\node (15) at (-5, 0.5) {};
				\node (16) at (-4.5, 0.5) {};
				\node (17) at (-2.5, 1.5) {};
				\node (18) at (-2.5, 1) {};
				\node (19) at (-1, 1.5) {};
				\node (20) at (-1, 1) {};
				\node (21) at (-3, 1.25) {};
				\node (22) at (-2.5, 1.25) {};
				\node (23) at (-3, 0.5) {};
				\node (24) at (-0.25, 0.5) {};
				\node (25) at (-1, 1.25) {};
				\node (26) at (-0.25, 1.25) {};
				\node (27) at (-3.75, 1.25) {$A_1$};
				\node (28) at (-3.75, 0.5) {$B_1$};
				\node (29) at (-1.75, 1.25) {$A_2$};
				\node (30) at (-4.5, 1.75) {};
				\node (31) at (-4.5, 0) {};
				\node (32) at (-4.5, -0.25) {$t_3$};
				\node (39) at (-1, 1.75) {};
				\node (40) at (-1, 0) {};
				\node (41) at (-1, -0.25) {$t_4$};
				\node (42) at (-5.5, 1.5) {};
				\node (43) at (-5.5, 0.25) {};
				\node (44) at (-5, 1.5) {};
				\node (45) at (-5, 0.25) {};
				\node (46) at (-6, 1.25) {};
				\node (47) at (-5.5, 1.25) {};
				\node (48) at (-6, 0.5) {};
				\node (49) at (-5.5, 0.5) {};
				\node (50) at (-5.25, 0.9) {$C_1$};
				\node (51) at (-5.5, 1.5) {};
				\node (52) at (-5.5, 1) {};
				\node (53) at (-5.5, 1.25) {};
				\node (54) at (-5.5, 1.75) {};
				\node (55) at (-5.5, 0) {};
				\node (56) at (-5.5, -0.25) {$t_1$};
				\node (57) at (-5, 1.5) {};
				\node (58) at (-5, 1) {};
				\node (59) at (-5, 1.25) {};
				\node (60) at (-5, 1.75) {};
				\node (61) at (-5, 0) {};
				\node (62) at (-5, -0.25) {$t_2$};
				\draw (0.center) to (1.center);
				\draw (0.center) to (2.center);
				\draw (1.center) to (4.center);
				\draw (2.center) to (4.center);
				\draw (9.center) to (10.center);
				\draw (9.center) to (11.center);
				\draw (10.center) to (12.center);
				\draw (11.center) to (12.center);
				\draw [thick] (13.center) to (14.center);
				\draw [thick] (15.center) to (16.center);
				\draw (17.center) to (18.center);
				\draw (17.center) to (19.center);
				\draw (18.center) to (20.center);
				\draw (19.center) to (20.center);
				\draw [thick] (21.center) to (22.center);
				\draw [->,thick] (25.center) to (26.center);
				\draw [->,thick] (23.center) to (24.center);
				\draw [dashed] (30.center) to (31.center);
				\draw [dashed] (39.center) to (40.center);
				\draw (42.center) to (43.center);
				\draw (42.center) to (44.center);
				\draw (43.center) to (45.center);
				\draw (44.center) to (45.center);
				\draw [thick] (46.center) to (47.center);
				\draw [thick] (48.center) to (49.center);
				\draw (51.center) to (52.center);
				\draw [dashed] (54.center) to (55.center);
				\draw (57.center) to (58.center);
				\draw [dashed] (60.center) to (61.center);
			\end{tikzpicture}
			\caption{Synchronous $S_2$}
			\label{fig:egmot-2a}
		\end{subfigure}
		\begin{subfigure}[b]{0.4\textwidth}
			\centering
			\begin{tikzpicture}
				\node (0) at (-4.5, 1.5) {};
				\node (1) at (-4.5, 1) {};
				\node (2) at (-3, 1.5) {};
				\node (4) at (-3, 1) {};
				\node (9) at (-3.5, 0.75) {};
				\node (10) at (-3.5, 0.25) {};
				\node (11) at (-2, 0.75) {};
				\node (12) at (-2, 0.25) {};
				\node (13) at (-5, 1.25) {};
				\node (14) at (-4.5, 1.25) {};
				\node (15) at (-5, 0.5) {};
				\node (16) at (-3.5, 0.5) {};
				\node (17) at (-2.5, 1.5) {};
				\node (18) at (-2.5, 1) {};
				\node (19) at (-1, 1.5) {};
				\node (20) at (-1, 1) {};
				\node (21) at (-3, 1.25) {};
				\node (22) at (-2.5, 1.25) {};
				\node (23) at (-2, 0.5) {};
				\node (24) at (-0.25, 0.5) {};
				\node (25) at (-1, 1.25) {};
				\node (26) at (-0.25, 1.25) {};
				\node (27) at (-3.75, 1.25) {$A_1$};
				\node (28) at (-2.75, 0.5) {$B_1$};
				\node (29) at (-1.75, 1.25) {$A_2$};
				\node (30) at (-4.5, 1.75) {};
				\node (31) at (-4.5, 0) {};
				\node (32) at (-4.5, -0.25) {$\tau_3$};
				\node (39) at (-1, 1.75) {};
				\node (40) at (-1, 0) {};
				\node (41) at (-1, -0.25) {$\tau_4$};
				\node (42) at (-5.5, 1.5) {};
				\node (43) at (-5.5, 0.25) {};
				\node (44) at (-5, 1.5) {};
				\node (45) at (-5, 0.25) {};
				\node (46) at (-6, 1.25) {};
				\node (47) at (-5.5, 1.25) {};
				\node (48) at (-6, 0.5) {};
				\node (49) at (-5.5, 0.5) {};
				\node (50) at (-5.25, 0.9) {$C_1$};
				\node (51) at (-5.5, 1.5) {};
				\node (52) at (-5.5, 1) {};
				\node (53) at (-5.5, 1.25) {};
				\node (54) at (-5.5, 1.75) {};
				\node (55) at (-5.5, 0) {};
				\node (56) at (-5.5, -0.25) {$\tau_1$};
				\node (57) at (-5, 1.5) {};
				\node (58) at (-5, 1) {};
				\node (59) at (-5, 1.25) {};
				\node (60) at (-5, 1.75) {};
				\node (61) at (-5, 0) {};
				\node (62) at (-5, -0.25) {$\tau_2$};
				\draw (0.center) to (1.center);
				\draw (0.center) to (2.center);
				\draw (1.center) to (4.center);
				\draw (2.center) to (4.center);
				\draw (9.center) to (10.center);
				\draw (9.center) to (11.center);
				\draw (10.center) to (12.center);
				\draw (11.center) to (12.center);
				\draw [thick] (13.center) to (14.center);
				\draw [thick] (15.center) to (16.center);
				\draw (17.center) to (18.center);
				\draw (17.center) to (19.center);
				\draw (18.center) to (20.center);
				\draw (19.center) to (20.center);
				\draw [thick] (21.center) to (22.center);
				\draw [->,thick] (25.center) to (26.center);
				\draw [->,thick] (23.center) to (24.center);
				\draw [dashed] (30.center) to (31.center);
				\draw [dashed] (39.center) to (40.center);
				\draw (42.center) to (43.center);
				\draw (42.center) to (44.center);
				\draw (43.center) to (45.center);
				\draw (44.center) to (45.center);
				\draw [thick] (46.center) to (47.center);
				\draw [thick] (48.center) to (49.center);
				\draw (51.center) to (52.center);
				\draw [dashed] (54.center) to (55.center);
				\draw (57.center) to (58.center);
				\draw [dashed] (60.center) to (61.center);
			\end{tikzpicture}
			\caption{Asynchronous $S_2'$}
			\label{fig:egmot-2b}
		\end{subfigure}
		\caption{Proving the synchronous system $S_2$ is equivalent to asynchronous $S_2'$.}
		\label{fig:egmot-2}
	\end{figure}
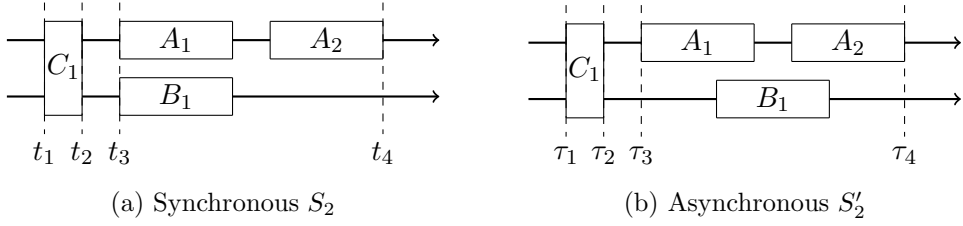
	Next we consider a harder example,
	which invalidates the reasoning in \Cref{eg:mot1}.
	In \Cref{fig:egmot-2a} is a synchronous system $S_2$,
	which compared to $S_1$,
	has an additional action $C_1$ on both qubits at the beginning.
	Specifically, we let $C_1$ performs a unitary quantum gate that prepares the entangled EPR pairs,
	i.e., $C_1\ket{xy}=\ket{\beta_{xy}}=\frac{1}{\sqrt{2}}\parens*{\ket{0}\ket{y}+(-1)^x\ket{1}\ket{1-y}}$ 
	for $x,y\in \braces*{0,1}$.
	Similarly we have the asynchronous system $S_2'$ in \Cref{fig:egmot-2b}.
 
	Problem will arise if we still use the reasoning in \Cref{eg:mot1},
	i.e., verify the equivalence 
	on all computational basis states.
	To see this, note that in $S_2$,
	starting from $\ket{\psi\parens*{t_1}}=\ket{xy}$,
	the state $\ket{\psi\parens*{t_2}}=\ket{\beta_{xy}}$ will be entangled.
	A similar statement holds for $S_2'$.
	Because $S_2$ (when restricted to the time region $[t_3,t_4]$) contains a similar structure to $S_1$,
	we met the same issue as in \Cref{eg:mot1}.
	We know $\ket{\psi\parens*{t_4}}=\parens*{A_2A_1\otimes B_1}\ket{\psi\parens*{t_3}}$ in $S_2$;
	but we do not know how to relate $\ket{\psi'\parens*{\tau_4}}$ and $\ket{\psi'\parens*{\tau_3}}$ in $S_2'$.

	There is still a way to circumvent the above new issue.
	The Hilbert space of two qubits is also spanned 
	by the EPR basis states $\braces*{\ket{\beta_{00}},\ket{\beta_{01}},\ket{\beta_{10}},\ket{\beta_{11}}}$.
	We can verify the equivalence between $S_2$ and $S_2'$
	on all these states.
	The merit is that in this case,
	$\ket{\psi(t_2)}=C_1\ket{\psi\parens*{t_1}}$ and $\ket{\psi'\parens*{\tau_2}}=C_1\ket{\psi'\parens*{\tau_1}}$ 
	will be transformed to product states by action $C_1$,
	and the issue disappears.
\end{example}

\begin{remark}
	\label{rmk:eg12-insight}
	Reviewing the above two examples, we obtain the following insights.
	The circumvention in \Cref{eg:mot1} can be generalised to show the equivalence is preserved under tensor product:
	if systems $G_1\simeq G_1'$ and $ G_2\simeq G_2'$,
	then $G_1\otimes G_2\simeq G_1'\otimes G_2'$.
	And the circumvention in \Cref{eg:mot2} can be generalised to show
	the equivalence is preserved under sequential composition:
	if systems $G_1\simeq G_1'$ and $G_2\simeq G_2'$,
	then $G_1\circ G_2\simeq G_1'\circ G_2'$.
\end{remark}

\begin{example}
	\label{eg:mot3}
	\begin{figure}
		\centering
		\begin{subfigure}[b]{0.4\textwidth}
			\centering
			\begin{tikzpicture}
				\node (0) at (-4.5, 0.75) {};
				\node (1) at (-4.5, -0.5) {};
				\node (2) at (-4, 0.75) {};
				\node (4) at (-4, -0.5) {};
				\node (9) at (-4.5, 1.5) {};
				\node (10) at (-4.5, 1) {};
				\node (11) at (-4, 1.5) {};
				\node (12) at (-4, 1) {};
				\node (17) at (-3.5, 0.75) {};
				\node (18) at (-3.5, 0.25) {};
				\node (19) at (-3, 0.75) {};
				\node (20) at (-3, 0.25) {};
				\node (21) at (-4, 0.5) {};
				\node (22) at (-3.5, 0.5) {};
				\node (23) at (-4, 1.25) {};
				\node (24) at (-2.25, 1.25) {};
				\node (25) at (-3, 0.5) {};
				\node (26) at (-2.25, 0.5) {};
				\node (27) at (-4.25, 0.15) {$C_3$};
				\node (28) at (-4.25, 1.25) {$A_2$};
				\node (29) at (-3.25, 0.5) {$B_3$};
				\node (39) at (-5.5, 1.75) {};
				\node (40) at (-5.5, -0.75) {};
				\node (41) at (-5.5, -1) {$t_1$};
				\node (42) at (-5.5, 1.5) {};
				\node (43) at (-5.5, 0.25) {};
				\node (44) at (-5, 1.5) {};
				\node (45) at (-5, 0.25) {};
				\node (46) at (-6, 1.25) {};
				\node (47) at (-5.5, 1.25) {};
				\node (48) at (-6, 0.5) {};
				\node (49) at (-5.5, 0.5) {};
				\node (50) at (-5.25, 0.9) {$C_2$};
				\node (51) at (-5.5, 1.5) {};
				\node (57) at (-5, 1.5) {};
				\node (63) at (-5.5, 0) {};
				\node (64) at (-5.5, -0.5) {};
				\node (65) at (-5, 0) {};
				\node (66) at (-5, -0.5) {};
				\node (67) at (-5.25, -0.25) {$B_2$};
				\node (68) at (-4, -0.25) {};
				\node (69) at (-2.25, -0.25) {};
				\node (70) at (-5, 1.25) {};
				\node (71) at (-4.5, 1.25) {};
				\node (72) at (-5, 0.5) {};
				\node (73) at (-4.5, 0.5) {};
				\node (74) at (-5, -0.25) {};
				\node (75) at (-4.5, -0.25) {};
				\node (76) at (-6, -0.25) {};
				\node (77) at (-5.5, -0.25) {};
				\node (78) at (-6.5, 0.75) {};
				\node (79) at (-6.5, -0.5) {};
				\node (80) at (-6, 0.75) {};
				\node (81) at (-6, -0.5) {};
				\node (82) at (-6.5, 1.5) {};
				\node (83) at (-6.5, 1) {};
				\node (84) at (-6, 1.5) {};
				\node (85) at (-6, 1) {};
				\node (86) at (-7.5, 0.75) {};
				\node (87) at (-7.5, 0.25) {};
				\node (88) at (-7, 0.75) {};
				\node (89) at (-7, 0.25) {};
				\node (91) at (-7.5, 0.5) {};
				\node (94) at (-7, 0.5) {};
				\node (96) at (-6.25, 0.15) {$C_1$};
				\node (97) at (-6.25, 1.25) {$A_1$};
				\node (98) at (-7.25, 0.5) {$B_1$};
				\node (101) at (-6.5, 1.25) {};
				\node (102) at (-6.5, 0.5) {};
				\node (103) at (-6.5, -0.25) {};
				\node (104) at (-8, 0.5) {};
				\node (106) at (-8, 1.25) {};
				\node (107) at (-8, -0.25) {};
				\node (108) at (-5, 1.75) {};
				\node (109) at (-5, -0.75) {};
				\node (110) at (-5, -1) {$t_2$};
				\node (111) at (-3, 1.75) {};
				\node (112) at (-3, -0.75) {};
				\node (113) at (-3, -1) {$t_3$};
				\draw (0.center) to (1.center);
				\draw (0.center) to (2.center);
				\draw (1.center) to (4.center);
				\draw (2.center) to (4.center);
				\draw (9.center) to (10.center);
				\draw (9.center) to (11.center);
				\draw (10.center) to (12.center);
				\draw (11.center) to (12.center);
				\draw (17.center) to (18.center);
				\draw (17.center) to (19.center);
				\draw (18.center) to (20.center);
				\draw (19.center) to (20.center);
				\draw [thick] (21.center) to (22.center);
				\draw [->,thick] (25.center) to (26.center);
				\draw [->,thick] (23.center) to (24.center);
				\draw [dashed] (39.center) to (40.center);
				\draw (42.center) to (43.center);
				\draw (42.center) to (44.center);
				\draw (43.center) to (45.center);
				\draw (44.center) to (45.center);
				\draw [thick] (46.center) to (47.center);
				\draw [thick] (48.center) to (49.center);
				\draw (63.center) to (64.center);
				\draw (63.center) to (65.center);
				\draw (64.center) to (66.center);
				\draw (65.center) to (66.center);
				\draw [->,thick] (68.center) to (69.center);
				\draw [thick] (70.center) to (71.center);
				\draw [thick] (72.center) to (73.center);
				\draw [thick] (74.center) to (75.center);
				\draw [thick] (76.center) to (77.center);
				\draw (78.center) to (79.center);
				\draw (78.center) to (80.center);
				\draw (79.center) to (81.center);
				\draw (80.center) to (81.center);
				\draw (82.center) to (83.center);
				\draw (82.center) to (84.center);
				\draw (83.center) to (85.center);
				\draw (84.center) to (85.center);
				\draw (86.center) to (87.center);
				\draw (86.center) to (88.center);
				\draw (87.center) to (89.center);
				\draw (88.center) to (89.center);
				\draw [thick] (94.center) to (102.center);
				\draw [thick] (106.center) to (101.center);
				\draw [thick] (104.center) to (91.center);
				\draw [thick] (107.center) to (103.center);
				\draw [dashed] (108.center) to (109.center);
				\draw [dashed] (111.center) to (112.center);
			\end{tikzpicture}
			\caption{Synchronous $S_3$}
			\label{fig:egmot-3a}
		\end{subfigure}
		\begin{subfigure}[b]{0.5\textwidth}
			\centering
			\begin{tikzpicture}
				\node (0) at (-3.5, 0.75) {};
				\node (1) at (-3.5, -0.5) {};
				\node (2) at (-3, 0.75) {};
				\node (4) at (-3, -0.5) {};
				\node (9) at (-4.25, 1.5) {};
				\node (10) at (-4.25, 1) {};
				\node (11) at (-2.25, 1.5) {};
				\node (12) at (-2.25, 1) {};
				\node (17) at (-2.75, 0.75) {};
				\node (18) at (-2.75, 0.25) {};
				\node (19) at (-2, 0.75) {};
				\node (20) at (-2, 0.25) {};
				\node (21) at (-3, 0.5) {};
				\node (22) at (-2.75, 0.5) {};
				\node (23) at (-2.25, 1.25) {};
				\node (24) at (-1.25, 1.25) {};
				\node (25) at (-2, 0.5) {};
				\node (26) at (-1.25, 0.5) {};
				\node (27) at (-3.25, 0.15) {$C_3$};
				\node (28) at (-3.23, 1.25) {$A_2$};
				\node (29) at (-2.35, 0.5) {$B_3$};
				\node (39) at (-5, 1.75) {};
				\node (40) at (-5, -0.75) {};
				\node (41) at (-5, -1) {$\tau_1$};
				\node (42) at (-5, 1.5) {};
				\node (43) at (-5, 0.25) {};
				\node (44) at (-4.5, 1.5) {};
				\node (45) at (-4.5, 0.25) {};
				\node (46) at (-5.25, 1.25) {};
				\node (47) at (-5, 1.25) {};
				\node (48) at (-6, 0.5) {};
				\node (49) at (-5, 0.5) {};
				\node (50) at (-4.75, 0.9) {$C_2$};
				\node (51) at (-5, 1.5) {};
				\node (57) at (-4.5, 1.5) {};
				\node (63) at (-5.75, 0) {};
				\node (64) at (-5.75, -0.5) {};
				\node (65) at (-3.75, 0) {};
				\node (66) at (-3.75, -0.5) {};
				\node (67) at (-4.75, -0.25) {$B_2$};
				\node (68) at (-3, -0.25) {};
				\node (69) at (-1.25, -0.25) {};
				\node (70) at (-4.5, 1.25) {};
				\node (71) at (-4.25, 1.25) {};
				\node (72) at (-4.5, 0.5) {};
				\node (73) at (-3.5, 0.5) {};
				\node (74) at (-3.75, -0.25) {};
				\node (75) at (-3.5, -0.25) {};
				\node (76) at (-6, -0.25) {};
				\node (77) at (-5.75, -0.25) {};
				\node (78) at (-6.5, 0.75) {};
				\node (79) at (-6.5, -0.5) {};
				\node (80) at (-6, 0.75) {};
				\node (81) at (-6, -0.5) {};
				\node (82) at (-7, 1.5) {};
				\node (83) at (-7, 1) {};
				\node (84) at (-5.25, 1.5) {};
				\node (85) at (-5.25, 1) {};
				\node (86) at (-7.5, 0.75) {};
				\node (87) at (-7.5, 0.25) {};
				\node (88) at (-6.75, 0.75) {};
				\node (89) at (-6.75, 0.25) {};
				\node (91) at (-7.5, 0.5) {};
				\node (94) at (-6.75, 0.5) {};
				\node (96) at (-6.25, 0.15) {$C_1$};
				\node (97) at (-6.12, 1.25) {$A_1$};
				\node (98) at (-7.12, 0.5) {$B_1$};
				\node (101) at (-7, 1.25) {};
				\node (102) at (-6.5, 0.5) {};
				\node (103) at (-6.5, -0.25) {};
				\node (104) at (-8, 0.5) {};
				\node (106) at (-8, 1.25) {};
				\node (107) at (-8, -0.25) {};
				\node (108) at (-4.5, 1.75) {};
				\node (109) at (-4.5, -0.75) {};
				\node (110) at (-4.5, -1) {$\tau_2$};
				\node (111) at (-2, 1.75) {};
				\node (112) at (-2, -0.75) {};
				\node (113) at (-2, -1) {$\tau_3$};
				\draw (0.center) to (1.center);
				\draw (0.center) to (2.center);
				\draw (1.center) to (4.center);
				\draw (2.center) to (4.center);
				\draw (9.center) to (10.center);
				\draw (9.center) to (11.center);
				\draw (10.center) to (12.center);
				\draw (11.center) to (12.center);
				\draw (17.center) to (18.center);
				\draw (17.center) to (19.center);
				\draw (18.center) to (20.center);
				\draw (19.center) to (20.center);
				\draw [thick] (21.center) to (22.center);
				\draw [->,thick] (25.center) to (26.center);
				\draw [->,thick] (23.center) to (24.center);
				\draw [dashed] (39.center) to (40.center);
				\draw (42.center) to (43.center);
				\draw (42.center) to (44.center);
				\draw (43.center) to (45.center);
				\draw (44.center) to (45.center);
				\draw [thick] (46.center) to (47.center);
				\draw [thick] (48.center) to (49.center);
				\draw (63.center) to (64.center);
				\draw (63.center) to (65.center);
				\draw (64.center) to (66.center);
				\draw (65.center) to (66.center);
				\draw [->,thick] (68.center) to (69.center);
				\draw [thick] (70.center) to (71.center);
				\draw [thick] (72.center) to (73.center);
				\draw [thick] (74.center) to (75.center);
				\draw [thick] (76.center) to (77.center);
				\draw (78.center) to (79.center);
				\draw (78.center) to (80.center);
				\draw (79.center) to (81.center);
				\draw (80.center) to (81.center);
				\draw (82.center) to (83.center);
				\draw (82.center) to (84.center);
				\draw (83.center) to (85.center);
				\draw (84.center) to (85.center);
				\draw (86.center) to (87.center);
				\draw (86.center) to (88.center);
				\draw (87.center) to (89.center);
				\draw (88.center) to (89.center);
				\draw [thick] (94.center) to (102.center);
				\draw [thick] (106.center) to (101.center);
				\draw [thick] (104.center) to (91.center);
				\draw [thick] (107.center) to (103.center);
				\draw [dashed] (108.center) to (109.center);
				\draw [dashed] (111.center) to (112.center);
			\end{tikzpicture}
			\caption{Asynchronous $S_3'$}
			\label{fig:egmot-3b}
		\end{subfigure}
		\caption{Proving the synchronous system $S_3$ is equivalent to asynchronous $S_3'$.}
		\label{fig:egmot-3}
	\end{figure}
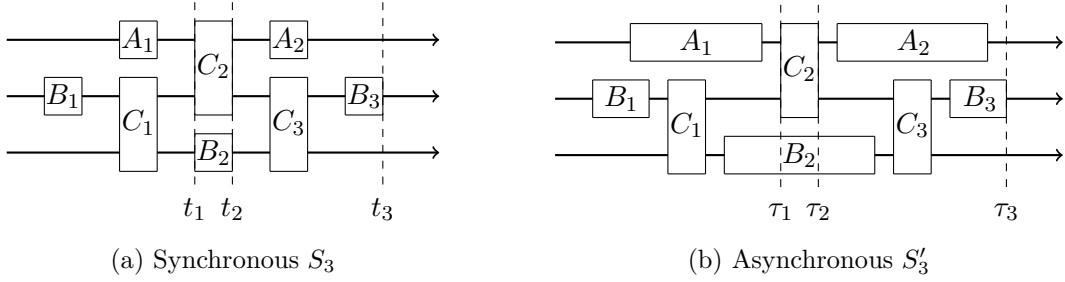
	Finally we examine an even  harder example, which  invalidates the reasoning in \Cref{eg:mot2}.
	In this example, $C_1,C_2$ and $C_3$
	are all EPR preparation unitary actions.
	In \Cref{fig:egmot-3a} is a synchronous system $S_3$,
	and in \Cref{fig:egmot-3b} is an asynchronous system $S_3'$.
	Note that our construction forces the entanglement to appear at some points in $S_3$ (and $S_3'$).
	In particular, one of the states $\ket{\psi\parens*{t_1}}$ and $\ket{\psi\parens*{t_2}}$ must be entangled
	because of $C_2$.
	A similar statement holds for $S_3'$.
	Since the structures of the systems are symmetric in time, w.l.o.g.,
	suppose $\ket{\psi\parens*{t_2}}$ is entangled.
	In this case, as $S_3$ (when restricted to the time region $[t_2,t_3]$)
	contains a similar structure to $S_1$,
	the issue in \Cref{eg:mot1} arises again.

	The circumventions in \Cref{eg:mot1,eg:mot2} no longer work.
	According to \Cref{rmk:eg12-insight},
	all we can try is to decompose $S_3$ and $S_3'$
	into tensor product or sequential composition of smaller systems whose equivalence are easier to prove,
	which is however impossible.
\end{example}

\subsection{Overview}
\label{sec:our_contributions}

\subsubsection{Identification of the challenges}
\label{sub:id-challenge}

From the examples in \Cref{sec:motivating_examples}, 
we observed the \textbf{first major challenge} caused by quantum entanglement.
It forces us to describe the state of a system as a (possibly entangled) whole (instead of a product),
which involves unknown implementation details of actions,
and consequently undermines our classical justification of the atomicity assumption.
Entanglement is actually a manifestation of the broader concept of quantum non-locality at the state level.
It is shown in~\cite{BDFMRSSW99} that quantum non-locality can also appear at the operation level:
there exist separable quantum operations that cannot be implemented by local actions.
Fortunately, the latter does not threat us 
because we are only concerned about the atomicity of local actions,
and we make no assumption about non-local actions 
(that may exhibit quantum non-locality at the operation level).

The \textbf{second major challenge} comes from the quantum measurement,
which is not illustrated in \Cref{sec:motivating_examples}.
Quantum measurement introduces probabilistic branching as another source of non-determinism.
A quantum process can go into any possible probabilistic branch,
and we have to consider the interaction between different branches from different processes.
This non-determinism is further complicated
by the notorious \textit{measurement problem}:
we do not know when the probabilistic branching actual occurs during a measurement!\footnote{
	Actually it is impossible to ask when the branching occurs,
	because the quantum measurement is not an instantaneous but 
	a continuous procedure, according to the decoherence theory~\cite{Schlosshauer05}.
}
Moreover, 
as most distributed systems of concern are non-terminating,
we need to properly define the probability of an event involving possibly infinitely many actions 
(e.g., for a quantum process repeatedly performing measurements, 
we may ask the probability of the outcome always being $1$).
Finally, we also need to handle the two sources of non-determinism ---
probabilistic branching and concurrency --- together harmoniously.

\subsubsection{Resolution of the challenges}

The first challenge caused by quantum entanglement 
is resolved by proper modeling of the distributed quantum systems 
and the corresponding dynamics.
The key insight is that we can actually describe the real-time state of a system at any time,
if we include all (effective) quantum degrees of freedom,
in particular, those implicit and uncontrollable ones in the quantum environments
introduced by measurements (see \Cref{rmk:quantum environments}).
This insight has not been exploited before in the literature of concurrency:
as mentioned in \Cref{rmk:con-Dijk-obser},
the classical state of a system at certain time is often undefined.

The second challenge caused by quantum measurement
is partly resolved by proper modeling of systems and the corresponding dynamics,
and partly resolved by techniques in proving the main theorem.
For the former,
we incorporate the non-determinism from probabilistic branching and concurrency in
a consistent model, 
and we carefully keep all definitions insensitive to the implementation details of actions.
As a result, our proofs do not depend on how quantum operations are actually performed,
thereby circumventing the measurement problem (see \Cref{rmk:circ-measurement-pr}).
For the latter,
to define and analyse the probabilities of events involving possibly infinitely many actions (due to non-termination),
we resort to the measure theoretical tools for probability.
(which are also used in previous works on probabilistic concurrency, e.g., \cite{Vardi85,SL94}).
We establish a natural connection between the state of a system and
the probabilities of classically observable events in the system,
via the system dynamics and observable dynamics discussed below.

\subsubsection{Non-atomic distributed quantum systems}
\label{sub:non-atomic-dis-q-sys}

As a basis for proving that local actions can be always regarded as atomic, we need a formal model in which non-atomic actions are allowed. So, our \textbf{first major contribution} is establishing 
such a model of non-atomic distributed quantum systems. Here, we informally introduce some basic ideas of the model. 
To develop this model, 
we first define actions, quantum processes and distributed quantum systems,
without assuming any atomicity.

\paragraph{Action}

An action is a set of bounded space-time events for performing a quantum operation.
It models a physical implementation (or execution) of an logical operation.
To an action one can assign several concerned properties.
Specifically, we are interested in the following properties of an action $a$.
\begin{enumerate}
	\item 
		$T[a]$: the time interval that $a$ spans.
	\item
		$q[a]$: the quantum register (i.e., a set of qubits) on which $a$ is supposed to perform.
	\item
		$\calE[a]$: the quantum operation that $a$ is supposed to perform,
		which is either a unitary quantum gate or a (partial) measurement.
	\item
		$e[a]$: the quantum environment (i.e., a set of quantum particles) introduced by $a$.
		If $\calE[a]$ is a unitary,
		then $e[a]=\emptyset$,
		because in this case $q[a]$ is effectively isolated from $e[a]$.
		If $\calE[a]$ is a (partial) measurement,
		then $e[a]$ includes the 
		quantum degrees of freedom in the measurement device.
\end{enumerate}
Note that in the above properties, 
$q[a]$ and $\calE[a]$ are specifications of the logical operation that $a$ performs,
while $T[a]$ and $e[a]$ are physical properties directly related to the physical reality that $a$ models.
Readers are referred to \Cref{sub:actions} for formal details.
Backgrounds of quantum states and quantum operations can be found in \Cref{sub:quantum_computing}.

\paragraph{Quantum process}

A \textit{quantum process} is a collection of countably many actions with a tree structure.
From an operational perspective, as the time progresses,
a quantum process repeatedly takes actions, one after another.
The tree structure comes from the probabilistic nature of quantum measurements:
performing a quantum measurement creates a probabilistic branching,
and in each branch is an action performing a partial measurement.
If we connect two successive actions with an directed edge $\rightarrow$,
then a quantum process can be represented by a rooted tree of actions,
with the following additional conditions:
\begin{enumerate}
	\item 
		(Sequentiality) For any two successive actions $a$ and $b$,
		$a$ appears before $b$ in time.
	\item
		(Branching) For a set of actions performing partial measurements
		from the same measurement,
		the quantum operations they perform should be consistent.
	\item
		(Finitely many actions in finite time)
		For any time $t$, there are finitely many actions that begin before $t$.
\end{enumerate}
Rigorous formulations of the above conditions can be found in \Cref{sub:quantum_process}.
Note that since a quantum process is a collection of actions,
it also models a physical implementation (execution).

\paragraph{Distributed quantum system}

A \textit{distributed quantum system} is a collection of multiple quantum processes.
We can also recursively define a distributed quantum system by the following rules:
\begin{enumerate}
	\item 
		Any single quantum process is a distributed quantum system.
	\item
		The parallel composition $A\parallel B$ of two distributed quantum systems $A$ and $B$
		with effectively separated quantum environments is also a distributed quantum system.
\end{enumerate}
Here, the quantum environment of a system $S$ 
is the union of quantum environment $e[a]$ introduced by every action $a$ in every quantum process in $S$.
The condition of effectively separated quantum environments
implies that quantum processes in a system can only communicate via their shared quantum registers 
(which are explicitly controllable objects),
but not via the quantum environments (which are uncontrollable objects).
Formal details of the above are presented in \Cref{sub:distributed_quantum_system}.

\paragraph{Local actions and atomicity}

After the notion of distributed quantum systems is defined,
we can describe what are local actions.
In a distributed quantum system $S$, an action $a$ is local,
if $a$ is space-time separated from actions in other processes in $S$;
that is, for any action $b$ in any other process in $S$,
we have either $q[a]\cap q[b]=\emptyset$ or $T[a]\cap T[b]=\emptyset$.
The formal definition of local actions is \Cref{def:local-act}.

We can also describe the notion of atomic actions in our model.
A subset $D$ of actions in a distributed quantum system are atomic,
if for any two actions $a$ and $b$ in $D$, either $a$ appears before $b$ in time,
or $b$ appears before $a$ (i.e., either $T[a]<T[b]$ or $T[b]<T[a]$).
The formal definition of atomicity is \Cref{def:atomicity}.

\paragraph{System dynamics}

Now we define the \textit{system dynamics} (see \Cref{def:sys-dyn}) of a distributed quantum system,
which characterises the real-time effect (semantics) of the system.
The real-time state of a system actually includes the explicit quantum states (of quantum registers and quantum environments)
and the implicit classical control states (internal to each quantum process, updating as the time progresses).
To handle them separately, we need the notions of partial processes and partial systems.

A quantum process $B$ is called a \textit{partial process} of another quantum process $A$,
if it consists of a rooted path to some action $b$ and the sub-tree rooted at $b$ in $A$.
Intuitively, $B$ is obtained from $A$ by knowing (i.e., fixing) the first several actions up to $b$,
among other probabilistic branches.
The knowledge of the first several actions corresponds to the implicit classical control state aforementioned.
Similarly, one can define a distributed quantum system to be a \textit{partial system} of another,
if each quantum process in the former is a partial process of a corresponding process in the latter.

To describe the real-time effect of a distributed quantum system $S$,
it suffices to define the real-time dynamics of the quantum state evolved according to any partial system of $S$.
Specifically, this is equivalent to define a function $\Bracks*{\cdot}\parens*{\cdot}$,
such that for any partial system $C$ and time $t$,
$\Bracks*{C}(t)$ is the quantum operation that maps the initial quantum state of the partial system $C$
to the quantum state at time $t$.
Note that due to the challenge caused by the measurement problem,
the real-time dynamics involves unknown physical details.
In particular, we do no even known when a measurement creates a probabilistic branching in real time.
Still, we can list the following conditions for the system dynamics $\Bracks*{\cdot}(\cdot)$ to satisfy:
\begin{enumerate}
	\item 
		(Initial condition) At time $0$, the system does nothing.
	\item
		(Branching) At any time after a branching in a process,
		the state evolved according to the partial system that contains all branches
		is a mix over states, each evolved according a partial system that contains exactly one branch.
	\item
		(Evolution) The evolution of the state (according to a partial system $C$)
		in a time interval $I$ is uniquely determined by all actions in $C$ that overlap with $I$,
		provided that the state corresponds to a single probabilistic branch at the beginning of $I$.

		Moreover, if $C$ is a parallel composition of two systems with no interaction in $I$,
		then the two systems evolve separately.
		Specifically, if a process in the first system takes a local action $a$ with time interval $I$,
		then the action correctly performs $\calE[a]$.
	\item
		(Trace) The trace\footnote{
			Here, the trace is the linear algebraic trace,
			which is widely used in quantum computing,
			but NOT the trace in the concurrency literature.
		}
		of the state evolved according to a partial system $C$ is non-increasing,
		and becomes a constant if $C$ is trace-preserving after some time.
\end{enumerate}

Formal details of these conditions are presented in \Cref{sub:system_dynamics}.
We make some further comments on the above conditions.
They are carefully kept insensitive to the implementation details of actions,
to circumvent the challenge caused by the measurement problem.
In particular, the condition (Branching) only specifies what happens after a branching (from quantum measurement),
but not during a branching (see also \Cref{rmk:circ-measurement-pr}).
The condition (Evolution) includes a formalisation of the Dijkstra-Lamport condition (see \Cref{sub:dijkstra_lamport_condition}),
which roughly says ``any local action is performed correctly''.
These conditions are natural and loose,
which can be satisfied by many possible functions $\Bracks*{\cdot}(\cdot)$.
We can arbitrarily pick one of them for analysis.

\paragraph{Observable dynamics}
The system dynamics describes the evolution of the quantum state according to partial systems.
However, what we are really concerned is the classically observable effect of a distributed quantum system,
because we as humans can only probe the quantum world via the quantum measurement.
Hence, from system dynamics, 
we further define the \textit{observable dynamics} (see \Cref{lmm:mu-is-pr-meas,def:equi-system}) of a distributed quantum system,
which characterises the probabilities of all classically observable events
in the system.
Observable dynamics can be thought of as the semantics of the system that one can observe.

Since the events in a distributed quantum system involve possibly infinitely many actions,
we need to use the measure theoretical tools (see \Cref{sub:probability_theory})
to define the probabilities of observable events.
Let $A$ be a quantum process, which has many probabilistic branches, 
As the time progresses, it only goes into one branch probabilistically,
which corresponds to a \textit{maximal path} in the tree structure of $A$.
Denote the set of all maximal paths in $A$ by $\omega(A)$.
Similarly, for a distributed quantum system $S$,
we can define $\omega(S)$ by taking the direct product of those sets of maximal paths in processes in $S$.

It turns out that the set $\braces*{\omega\parens*{C}:\textup{$C$ is a partial system of $S$}}$
generates (in the sense of $\sigma$-algebra; see \Cref{sub:probability_theory})
the set of observable events of concern.
Given an initial quantum state $\rho$, we can define a consistent probability measure $\mu_{\rho\rightarrow S}$ for observable events,
according to the system dynamics of $S$.
Specifically, the probability of a partial system $C$ is exactly the trace of the quantum state evolved according to $C$
in the limit $t\to +\infty$.
The function $\mu_{\cdot \rightarrow S}$ then completely captures the observable dynamics of $S$.

Consider two distributed quantum systems $S_1$ and $S_2$.
If they correspond to the same logical system;
that is, roughly speaking, they are the same when abstracting out the physical properties 
($T[a]$ and $e[a]$ for every action $a$ within),
then we can construct an isomorphism between $S_1$ and $S_2$.
Further, if their observable dynamics coincide, then they are defined to be \textit{equivalent}, denoted by $S_1\simeq S_2$.
That is, they are indistinguishable to any classical observer.
Formal details of the above are presented in \Cref{sub:observable_dynamics}.

\subsubsection{Local actions regarded atomic}

Based upon the established model,
our \textbf{second major contribution} is the following theorem,
serving as a rigorous basis for the atomicity of local actions (see \Cref{def:local-act}).

\begin{theorem}[Informal version of \Cref{thm:local-atom}]
	\label{thm:main-inform}
	For any (physically implementable)\footnote{
		See \Cref{thm:local-atom} for details.
	}
	distributed quantum system $S$,
	there is an equivalent system $S'\simeq S$,
	such that local actions in $S'$ are atomic.
\end{theorem}

In other words, in a distributed quantum system,
one can safely assume the atomicity of local actions,
if only the observable dynamics is of concern.
From the definition of observable dynamics, 
it contains all information that a classical observer can extract from the system,
and is all an external programmer needs.
The statement in \Cref{thm:main-inform} does not generally hold if the observable dynamics is replaced by the system dynamics.
Formal details of the theorem are presented in \Cref{sec:atomicity_of_local_actions}.

It is also worth pointing out that even a system only consists of local actions,
the state of the system can still be entangled (see examples in \Cref{sec:motivating_examples}).

\paragraph{Application to motivating examples}

Indeed, \Cref{thm:main-inform} can be applied to solve the problem considered in \Cref{sec:motivating_examples}.
Consider adding additional single-qubit quantum measurement actions $D,E,F$ in the systems in \Cref{eg:mot3},
as shown in \Cref{fig:egmot-4}.
Using \Cref{thm:main-inform}, $S_4$ and $S_4'$,
containing only local actions, have the same observable dynamics,
because it is easy to verify their atomic versions are the same.
As $D,E,F$ can be arbitrarily chosen,
by some basic properties of quantum operations,
it immediately follows that $\ket{\psi\parens*{t}}=\ket{\psi'\parens*{\tau}}$, as desired.

\begin{figure}
	\centering
	\begin{subfigure}[b]{0.4\textwidth}
		\centering
		\begin{tikzpicture}
			\node (0) at (-4.25, 0.75) {};
			\node (1) at (-4.25, -0.5) {};
			\node (2) at (-3.75, 0.75) {};
			\node (4) at (-3.75, -0.5) {};
			\node (9) at (-4.25, 1.5) {};
			\node (10) at (-4.25, 1) {};
			\node (11) at (-3.75, 1.5) {};
			\node (12) at (-3.75, 1) {};
			\node (17) at (-3.5, 0.75) {};
			\node (18) at (-3.5, 0.25) {};
			\node (19) at (-3, 0.75) {};
			\node (20) at (-3, 0.25) {};
			\node (21) at (-3.75, 0.5) {};
			\node (22) at (-3.5, 0.5) {};
			\node (23) at (-3.75, 1.25) {};
			\node (25) at (-3, 0.5) {};
			\node (27) at (-4, 0.15) {$C_3$};
			\node (28) at (-4, 1.25) {$A_2$};
			\node (29) at (-3.25, 0.5) {$B_3$};
			\node (42) at (-5, 1.5) {};
			\node (43) at (-5, 0.25) {};
			\node (44) at (-4.5, 1.5) {};
			\node (45) at (-4.5, 0.25) {};
			\node (46) at (-5.25, 1.25) {};
			\node (47) at (-5, 1.25) {};
			\node (48) at (-5.25, 0.5) {};
			\node (49) at (-5, 0.5) {};
			\node (50) at (-4.75, 0.9) {$C_2$};
			\node (51) at (-5, 1.5) {};
			\node (57) at (-4.5, 1.5) {};
			\node (63) at (-5, 0) {};
			\node (64) at (-5, -0.5) {};
			\node (65) at (-4.5, 0) {};
			\node (66) at (-4.5, -0.5) {};
			\node (67) at (-4.75, -0.25) {$B_2$};
			\node (68) at (-3.75, -0.25) {};
			\node (70) at (-4.5, 1.25) {};
			\node (71) at (-4.25, 1.25) {};
			\node (72) at (-4.5, 0.5) {};
			\node (73) at (-4.25, 0.5) {};
			\node (74) at (-4.5, -0.25) {};
			\node (75) at (-4.25, -0.25) {};
			\node (76) at (-5.25, -0.25) {};
			\node (77) at (-5, -0.25) {};
			\node (78) at (-5.75, 0.75) {};
			\node (79) at (-5.75, -0.5) {};
			\node (80) at (-5.25, 0.75) {};
			\node (81) at (-5.25, -0.5) {};
			\node (82) at (-5.75, 1.5) {};
			\node (83) at (-5.75, 1) {};
			\node (84) at (-5.25, 1.5) {};
			\node (85) at (-5.25, 1) {};
			\node (86) at (-6.5, 0.75) {};
			\node (87) at (-6.5, 0.25) {};
			\node (88) at (-6, 0.75) {};
			\node (89) at (-6, 0.25) {};
			\node (91) at (-6.5, 0.5) {};
			\node (94) at (-6, 0.5) {};
			\node (96) at (-5.5, 0.15) {$C_1$};
			\node (97) at (-5.5, 1.25) {$A_1$};
			\node (98) at (-6.25, 0.5) {$B_1$};
			\node (101) at (-5.75, 1.25) {};
			\node (102) at (-5.75, 0.5) {};
			\node (103) at (-5.75, -0.25) {};
			\node (104) at (-6.75, 0.5) {};
			\node (106) at (-6.75, 1.25) {};
			\node (107) at (-6.75, -0.25) {};
			\node (111) at (-2.75, 1.75) {};
			\node (112) at (-2.75, -0.75) {};
			\node (113) at (-2.75, -1) {$t$};
			\node (114) at (-2.75, 0.75) {};
			\node (115) at (-2.75, 0.25) {};
			\node (116) at (-2.25, 0.75) {};
			\node (117) at (-2.25, 0.25) {};
			\node (118) at (-2.75, 0.5) {};
			\node (119) at (-2.25, 0.5) {};
			\node (120) at (-2.5, 0.5) {$E$};
			\node (121) at (-2.75, 1.5) {};
			\node (122) at (-2.75, 1) {};
			\node (123) at (-2.25, 1.5) {};
			\node (124) at (-2.25, 1) {};
			\node (125) at (-2.75, 1.25) {};
			\node (126) at (-2.25, 1.25) {};
			\node (127) at (-2.5, 1.25) {$D$};
			\node (128) at (-2.75, 0) {};
			\node (129) at (-2.75, -0.5) {};
			\node (130) at (-2.25, 0) {};
			\node (131) at (-2.25, -0.5) {};
			\node (132) at (-2.75, -0.25) {};
			\node (133) at (-2.25, -0.25) {};
			\node (134) at (-2.5, -0.25) {$F$};
			\node (135) at (-1.5, 1.25) {};
			\node (137) at (-1.5, 0.5) {};
			\node (138) at (-1.5, -0.25) {};
			\draw (0.center) to (1.center);
			\draw (0.center) to (2.center);
			\draw (1.center) to (4.center);
			\draw (2.center) to (4.center);
			\draw (9.center) to (10.center);
			\draw (9.center) to (11.center);
			\draw (10.center) to (12.center);
			\draw (11.center) to (12.center);
			\draw (17.center) to (18.center);
			\draw (17.center) to (19.center);
			\draw (18.center) to (20.center);
			\draw (19.center) to (20.center);
			\draw [thick] (21.center) to (22.center);
			\draw (42.center) to (43.center);
			\draw (42.center) to (44.center);
			\draw (43.center) to (45.center);
			\draw (44.center) to (45.center);
			\draw [thick] (46.center) to (47.center);
			\draw [thick] (48.center) to (49.center);
			\draw (63.center) to (64.center);
			\draw (63.center) to (65.center);
			\draw (64.center) to (66.center);
			\draw (65.center) to (66.center);
			\draw [thick] (70.center) to (71.center);
			\draw [thick] (72.center) to (73.center);
			\draw [thick] (74.center) to (75.center);
			\draw [thick] (76.center) to (77.center);
			\draw (78.center) to (79.center);
			\draw (78.center) to (80.center);
			\draw (79.center) to (81.center);
			\draw (80.center) to (81.center);
			\draw (82.center) to (83.center);
			\draw (82.center) to (84.center);
			\draw (83.center) to (85.center);
			\draw (84.center) to (85.center);
			\draw (86.center) to (87.center);
			\draw (86.center) to (88.center);
			\draw (87.center) to (89.center);
			\draw (88.center) to (89.center);
			\draw [thick] (94.center) to (102.center);
			\draw [thick] (106.center) to (101.center);
			\draw [thick] (104.center) to (91.center);
			\draw [thick] (107.center) to (103.center);
			\draw [dashed] (111.center) to (112.center);
			\draw (114.center) to (115.center);
			\draw (114.center) to (116.center);
			\draw (115.center) to (117.center);
			\draw (116.center) to (117.center);
			\draw [thick] (25.center) to (118.center);
			\draw (121.center) to (122.center);
			\draw (121.center) to (123.center);
			\draw (122.center) to (124.center);
			\draw (123.center) to (124.center);
			\draw (128.center) to (129.center);
			\draw (128.center) to (130.center);
			\draw (129.center) to (131.center);
			\draw (130.center) to (131.center);
			\draw [thick] (23.center) to (125.center);
			\draw [thick] (68.center) to (132.center);
			\draw [->,thick] (126.center) to (135.center);
			\draw [->,thick] (119.center) to (137.center);
			\draw [->,thick] (133.center) to (138.center);
		\end{tikzpicture}
		\caption{Synchronous $S_4$}
		\label{fig:egmot-4a}
	\end{subfigure}
	\begin{subfigure}[b]{0.4\textwidth}
		\centering
		\begin{tikzpicture}
			\node (0) at (-3.75, 0.75) {};
			\node (1) at (-3.75, -0.5) {};
			\node (2) at (-3.25, 0.75) {};
			\node (4) at (-3.25, -0.5) {};
			\node (9) at (-4.25, 1.5) {};
			\node (10) at (-4.25, 1) {};
			\node (11) at (-2.75, 1.5) {};
			\node (12) at (-2.75, 1) {};
			\node (17) at (-3, 0.75) {};
			\node (18) at (-3, 0.25) {};
			\node (19) at (-2.5, 0.75) {};
			\node (20) at (-2.5, 0.25) {};
			\node (21) at (-3.25, 0.5) {};
			\node (22) at (-3, 0.5) {};
			\node (23) at (-2.75, 1.25) {};
			\node (25) at (-2.5, 0.5) {};
			\node (27) at (-3.5, 0.15) {$C_3$};
			\node (28) at (-3.5, 1.25) {$A_2$};
			\node (29) at (-2.75, 0.5) {$B_3$};
			\node (42) at (-5, 1.5) {};
			\node (43) at (-5, 0.25) {};
			\node (44) at (-4.5, 1.5) {};
			\node (45) at (-4.5, 0.25) {};
			\node (46) at (-5.25, 1.25) {};
			\node (47) at (-5, 1.25) {};
			\node (48) at (-5.75, 0.5) {};
			\node (49) at (-5, 0.5) {};
			\node (50) at (-4.75, 0.9) {$C_2$};
			\node (51) at (-5, 1.5) {};
			\node (57) at (-4.5, 1.5) {};
			\node (63) at (-5.5, 0) {};
			\node (64) at (-5.5, -0.5) {};
			\node (65) at (-4, 0) {};
			\node (66) at (-4, -0.5) {};
			\node (67) at (-4.75, -0.25) {$B_2$};
			\node (68) at (-3.25, -0.25) {};
			\node (70) at (-4.5, 1.25) {};
			\node (71) at (-4.25, 1.25) {};
			\node (72) at (-4.5, 0.5) {};
			\node (73) at (-3.75, 0.5) {};
			\node (74) at (-4, -0.25) {};
			\node (75) at (-3.75, -0.25) {};
			\node (76) at (-5.75, -0.25) {};
			\node (77) at (-5.5, -0.25) {};
			\node (78) at (-6.25, 0.75) {};
			\node (79) at (-6.25, -0.5) {};
			\node (80) at (-5.75, 0.75) {};
			\node (81) at (-5.75, -0.5) {};
			\node (82) at (-6.75, 1.5) {};
			\node (83) at (-6.75, 1) {};
			\node (84) at (-5.25, 1.5) {};
			\node (85) at (-5.25, 1) {};
			\node (86) at (-7, 0.75) {};
			\node (87) at (-7, 0.25) {};
			\node (88) at (-6.5, 0.75) {};
			\node (89) at (-6.5, 0.25) {};
			\node (91) at (-7, 0.5) {};
			\node (94) at (-6.5, 0.5) {};
			\node (96) at (-6, 0.15) {$C_1$};
			\node (97) at (-6, 1.25) {$A_1$};
			\node (98) at (-6.75, 0.5) {$B_1$};
			\node (101) at (-6.75, 1.25) {};
			\node (102) at (-6.25, 0.5) {};
			\node (103) at (-6.25, -0.25) {};
			\node (104) at (-7.5, 0.5) {};
			\node (106) at (-7.5, 1.25) {};
			\node (107) at (-7.5, -0.25) {};
			\node (111) at (-2.25, 1.75) {};
			\node (112) at (-2.25, -0.75) {};
			\node (113) at (-2.25, -1) {$\tau$};
			\node (114) at (-2.25, 0.75) {};
			\node (115) at (-2.25, 0.25) {};
			\node (116) at (-1.75, 0.75) {};
			\node (117) at (-1.75, 0.25) {};
			\node (118) at (-2.25, 0.5) {};
			\node (119) at (-1.75, 0.5) {};
			\node (120) at (-2, 0.5) {$E$};
			\node (121) at (-2.25, 0) {};
			\node (122) at (-2.25, -0.5) {};
			\node (123) at (-1.75, 0) {};
			\node (124) at (-1.75, -0.5) {};
			\node (125) at (-2.25, -0.25) {};
			\node (126) at (-1.75, -0.25) {};
			\node (127) at (-2, -0.25) {$F$};
			\node (128) at (-2.25, 1.5) {};
			\node (129) at (-2.25, 1) {};
			\node (130) at (-1.75, 1.5) {};
			\node (131) at (-1.75, 1) {};
			\node (132) at (-2.25, 1.25) {};
			\node (133) at (-1.75, 1.25) {};
			\node (134) at (-2, 1.25) {$D$};
			\node (135) at (-1, 1.25) {};
			\node (136) at (-1, 0.5) {};
			\node (137) at (-1, -0.25) {};
			\draw (0.center) to (1.center);
			\draw (0.center) to (2.center);
			\draw (1.center) to (4.center);
			\draw (2.center) to (4.center);
			\draw (9.center) to (10.center);
			\draw (9.center) to (11.center);
			\draw (10.center) to (12.center);
			\draw (11.center) to (12.center);
			\draw (17.center) to (18.center);
			\draw (17.center) to (19.center);
			\draw (18.center) to (20.center);
			\draw (19.center) to (20.center);
			\draw [thick] (21.center) to (22.center);
			\draw (42.center) to (43.center);
			\draw (42.center) to (44.center);
			\draw (43.center) to (45.center);
			\draw (44.center) to (45.center);
			\draw [thick] (46.center) to (47.center);
			\draw [thick] (48.center) to (49.center);
			\draw (63.center) to (64.center);
			\draw (63.center) to (65.center);
			\draw (64.center) to (66.center);
			\draw (65.center) to (66.center);
			\draw [thick] (70.center) to (71.center);
			\draw [thick] (72.center) to (73.center);
			\draw [thick] (74.center) to (75.center);
			\draw [thick] (76.center) to (77.center);
			\draw (78.center) to (79.center);
			\draw (78.center) to (80.center);
			\draw (79.center) to (81.center);
			\draw (80.center) to (81.center);
			\draw (82.center) to (83.center);
			\draw (82.center) to (84.center);
			\draw (83.center) to (85.center);
			\draw (84.center) to (85.center);
			\draw (86.center) to (87.center);
			\draw (86.center) to (88.center);
			\draw (87.center) to (89.center);
			\draw (88.center) to (89.center);
			\draw [thick] (94.center) to (102.center);
			\draw [thick] (106.center) to (101.center);
			\draw [thick] (104.center) to (91.center);
			\draw [thick] (107.center) to (103.center);
			\draw [dashed] (111.center) to (112.center);
			\draw (114.center) to (115.center);
			\draw (114.center) to (116.center);
			\draw (115.center) to (117.center);
			\draw (116.center) to (117.center);
			\draw [thick] (25.center) to (118.center);
			\draw (121.center) to (122.center);
			\draw (121.center) to (123.center);
			\draw (122.center) to (124.center);
			\draw (123.center) to (124.center);
			\draw (128.center) to (129.center);
			\draw (128.center) to (130.center);
			\draw (129.center) to (131.center);
			\draw (130.center) to (131.center);
			\draw [thick] (23.center) to (132.center);
			\draw [thick] (68.center) to (125.center);
			\draw [->,thick] (133.center) to (135.center);
			\draw [->,thick] (119.center) to (136.center);
			\draw [->,thick] (126.center) to (137.center);
		\end{tikzpicture}
		\caption{Asynchronous $S_4'$}
		\label{fig:egmot-4b}
	\end{subfigure}
	\caption{Appending systems in \Cref{fig:egmot-3} with extra single-qubit quantum measurement actions.}
	\label{fig:egmot-4}
\end{figure}
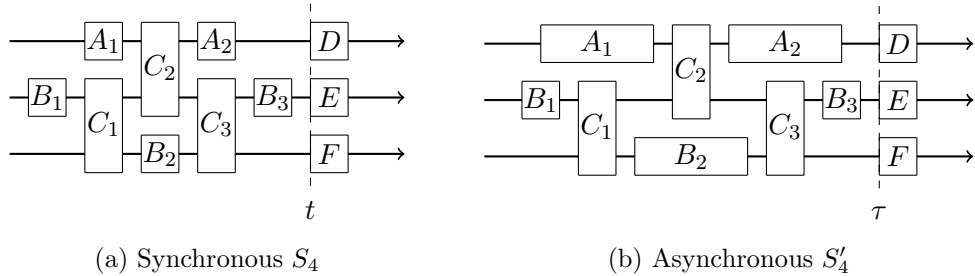

\subsection{Related Works}
\label{sub:related_works}

The most related classical works are by Lamport on non-atomic systems.
He proposed two ways of reasoning about non-atomic systems.
The first is behaviour reasoning, starting from \cite{Lamport74},
and finally developed into the two-arrow model~\cite{Lamport79,Lamport86a,Lamport86b,Lamport86c,Lamport86d}.
The two-arrow model is concerned about two temporal relations (partial orders) of actions: precedence and causality.\footnote{
	The terminologies in~\cite{Lamport86a} are different from this paper.
	In particular, action and system in this paper correspond to 
	operation execution and system execution in~\cite{Lamport86a}, respectively.
}
By a series of axioms, 
he proved the correctness of solutions to the mutual exclusion problem~\cite{Lamport79,Lamport86a,Lamport86b}
and the constructions of atomic registers~\cite{Lamport86c,Lamport86d}.
The second is assertional reasoning, 
which is based on the more elementary notions of states and actions,\footnote{
	Note that temporal relations are abstracted from the states,
	and omit certain information.
	For example, two non-atomic systems may have the same set of temporal relations
	but different sets of states.
}
and more formal than behaviour reasoning.
Using assertional reasoning, 
two unrevealed assumptions in previous correctness proofs~\cite{Lamport74,Lamport77,Lamport79}
of the bakery algorithm~\cite{Lamport74} are discovered~\cite{Lamport90b}.
Later works on atomicity further considered other desirable properties in concurrency.
For example, Herlihy and Wing considered linearizability~\cite{HW90} by taking into account the real-time orders in addition to the atomicity,
which has been later widely studied (e.g., \cite{Neiger94,CRR18}).

In comparison, our reasoning in \Cref{sec:a_model_of_distributed_quantum_system,sec:atomicity_of_local_actions}
is also based on states (and NOT on partial orders),
but we are concerned about real-time physical states instead of discrete abstract states (see the comparison in \Cref{rmk:con-Dijk-obser}).
The former are more basic in quantum computing,
because they directly involve quantum entanglement,
and quantum evolution is naturally continuous-time.

\subsection{Discussion}
\label{sec:discussion}

In this paper we point out the importance and non-triviality of justifying the atomicity assumption in distributed quantum systems,
in particular, via a series of motivating examples.
We identify the challenges caused by quantum entanglement and measurement,
which are then resolved by proper modeling of the system
and non-trivial techniques in the proofs.
Specifically, we establish a formal model of non-atomic distributed quantum systems,
upon which we prove that local actions can be regarded as if they were atomic
up to the observable dynamics of the system.
This provides a rigorous guarantee for assuming the atomicity of local actions.

This paper is just one of the first steps to a theory of concurrency in quantum computing.
We list several questions for future works as follows.
\begin{enumerate}
	\item 
		\label{stp:open-q-1}
		We have rigorously justified the atomicity of local actions in distributed quantum computing.
		How about the atomicity of non-local actions?
		As mentioned in~\Cref{sec:introduction},
		in practice, it can be guaranteed at the classical level.
		Can this be done by pure quantum software (like in the classical case~\cite{Lamport86c,Lamport86d})?
		This problem is of both philosophical and physical interests.
		Although quantum computing generalises classical computing,
		an action in classical computing (e.g., read/write)
		is \textit{not} a \textit{single} action in quantum computing (i.e., unitary and measurement).
		Indeed, a classical action consists of multiple quantum actions,
		which implies classical methods~\cite{Lamport86c,Lamport86d}
		do not directly carry over to the quantum case.
	\item
		The mutual exclusion problem is closely related to the atomicity,
		and of great importance in concurrency.
		Can this problem be solved by pure quantum software?
		Due to similar reasons for Question~\ref{stp:open-q-1},
		classical results~\cite{Lamport86a,Lamport86b} do not directly carry over to the quantum case.
	\item
		In this paper, systems considered have classical control flow.
        That is, implementing a quantum process implicitly requires a classical internal state
        updating as the process progresses in time.
        The separation of this internal state and the quantum state is captured by \Cref{def:partial-proc}.
        Quantum control flows are also considered in the literature~\cite{Ying16}.
        What is the concurrency with quantum control flow?
		It seems to allow a quantum superposition of different temporal orders between actions.
        We reserve for it the name ``quantum concurrency'',
        and leave it for future investigation.
	\item
        Verifying sequential quantum systems is of great practical interest (e.g.\ recently, \cite{CCLLTY23,BLS23}).
        There are also attempts to verify high-level concurrent/distributed quantum systems (e.g., \cite{FLY22,YZLF22}).
        How to verify properties of low-level distributed quantum systems (like the one in this paper)?
\end{enumerate}

\section{Structure of the Paper}
\label{sec:structure_of_the_paper}

The remainder of this paper is structured as follows.
In \Cref{sec:backgrounds}, the preliminaries of quantum computing and probability theory are presented.
In \Cref{sec:a_model_of_distributed_quantum_system}, we develop our model of non-atomic distributed quantum systems
step by step.
We first define actions in \Cref{sub:actions},
quantum processes in \Cref{sub:quantum_process},
and distributed quantum systems in \Cref{sub:distributed_quantum_system}.
Based on these concepts,
we define the system dynamics of a distributed quantum system in \Cref{sub:system_dynamics}.
Further, in \Cref{sub:observable_dynamics}, we define the observable dynamics of a distributed quantum system,
induced by the system dynamics.
Finally, in \Cref{sec:atomicity_of_local_actions},
we prove our main theorems that justify the atomicity assumption of local actions.
For readability, all formal proofs are deferred to \Cref{sec:details_of_proofs},
and proof sketches are provided for the theorems in \Cref{sec:atomicity_of_local_actions}.

\section{Preliminaries}
\label{sec:backgrounds}

\subsection{Quantum Computing}
\label{sub:quantum_computing}

In this section we briefly introduce the basic notions and notations in quantum computing.
For a more detailed introduction, the readers are referred to the textbook~\cite{NC10}.

\paragraph{Quantum states}
The state space of a closed quantum physical system is a Hilbert space,
in which a quantum state can be represented by a unit complex vector.
We use Dirac's notation $\ket{\psi},\ket{\phi},\ldots$ to denote vectors.
The inner product of two vectors $\ket{\psi}$ and $\ket{\phi}$ is denoted by $\braket{\psi|\phi}$.
By definition any quantum state $\ket{\psi}$ has norm $\norm*{\psi}=\sqrt{\braket{\psi|\psi}}=1$.
The Hilbert space can be discrete or continuous, 
depending on the physical system of concern.
Qubits are elementary controllable objects in quantum computing.
A qubit has a discrete Hilbert space $\Co^2$ of dimension $2$.
The Hilbert space of a composite quantum physical system is the tensor product of those of the components,
e.g., the Hilbert space of $n$ qubits is $\parens*{\Co^2}^{\otimes n}\simeq\Co^{2^n}$.
For a discrete Hilbert space $\calH$ of finite dimension $d$,
one can pick an orthonormal basis $\braces*{\ket{j}}_{j=0}^{d-1}$ with $\ket{j}\in \calH$,
and choose it to be the computational basis.
Then, any quantum state $\ket{\psi}\in \calH$ can be written as $\ket{\psi}=\sum_{j} \alpha_j \ket{j}$
with $\sum_{j}\abs*{\alpha_j}^2=1$, a superposition of all these computational basis states.
Quantum superposition leads to the phenomenon of quantum entanglement.
If a quantum state $\ket{\psi}$ can be written as a product $\ket{\psi_1}\otimes \ket{\psi_2}$ 
(abbreviated as $\ket{\psi_1}\ket{\psi_2}$),
then it is a product state;
otherwise it is an entangled state.

The above states in a Hilbert space are pure states.
When probability is introduced (e.g., by quantum measurements),
we need to consider mixed quantum states.
A positive semi-definite operator $\rho$ with trace in $[0,1]$ is called a density operator;
i.e., for any $\ket{x}\in \calH$, $\bra{x}\rho\ket{x}\geq 0$ and $\tr\parens*{\rho}\in [0,1]$.
A mixed quantum state is represented by a density operator $\rho$ with $\tr\parens*{\rho}=1$.
The unit trace represents full knowledge of the mixed quantum state.
Sometimes we will slightly abuse the terminology by calling 
a density operator (without the unit trace condition) a quantum state.
Any pure quantum state $\ket{\psi}$ has a corresponding density operator $\ket{\psi}\!\bra{\psi}$.
Any mixed quantum state $\rho$ can also be represented as an ensemble of pure quantum states,
via the spectral decomposition $\rho=\sum_{j} p_j\ket{\psi_j}\!\bra{\psi_j}$,
of which an interpretation is with probability $p_j>0$ it is in state $\ket{\psi_j}$.
Note that when $\tr\parens*{\rho}=1$, the sum of probabilities $\sum_{j}p_j=1$.
For a Hilbert space $\calH$,
we use $\calD\parens*{\calH}$ to denote the set of density operators on $\calH$.

An example of density operator is as follows.
\begin{equation*}
	\rho=\frac{1}{4}\ket{0}\!\bra{0}+\frac{1}{2}\ket{+}\!\bra{+}
	=
	\begin{bmatrix}
		0.5&0.25\\
		0.25&0.25
	\end{bmatrix},
\end{equation*}
where $\ket{+}=\frac{1}{\sqrt{2}}\parens*{\ket{0}+\ket{1}}$.

\paragraph{Quantum operations}
In quantum computing, we can perform two types of quantum operations.
The first is quantum gate, which can be represented by a unitary operator $U$ with $UU^\dagger =U^\dagger U=\Id$.
After applying a quantum gate $U$ on a pure state $\ket{\psi}$,
one will obtain the state $U\ket{\psi}$.
We give several examples as follows.
An $X$ gate is a single-qubit gate such that $X\ket{0}=\ket{1}$ and $X\ket{1}=\ket{0}$.
A Hadamard gate $H$ is also a single qubit gate such that $H\ket{0}=\frac{1}{\sqrt{2}}\parens*{\ket{0}+\ket{1}}$
and $H\ket{1}=\frac{1}{\sqrt{2}}\parens*{\ket{0}-\ket{1}}$.
A \textit{CNOT} gate is a two-qubit gate, represented by $\mathit{CNOT}=\ket{0}\!\bra{0}\otimes \Id+\ket{1}\!\bra{1}\otimes X$.
The EPR preparation unitary mentioned in \Cref{sec:motivating_examples} is a two-qubit gate,
represented by $\parens*{H\otimes \Id}\mathit{CNOT}$.

The second is quantum measurement,
which can extract classical information from the quantum state.
A measurement can be represented by a set of operators $\braces*{M_m}_m$
such that $\sum_{m}M_m^\dagger M_m=\Id$.
After applying the measurement to a pure state $\ket{\psi}$,
one will obtain the classical outcome $m$,
and the corresponding state $M_m\ket{\psi}/\norm*{M_m\ket{\psi}}$
with probability $\norm*{M_m\ket{\psi}}^2$.
An example of measurement is in the computational basis,
with $M_m=\ket{m}\!\bra{m}$.

Quantum gates and quantum measurements can be described in the unified framework of quantum operations on density operators.
A general quantum operation $\calE$ on Hilbert space $\calH$ 
is a completely positive non-trace-increasing map on $\calD\parens*{\calH}$.
Equivalently, $\calE$ can be represented by a set of Kraus operators $\braces*{E_k}_k$
such that $0\sqsubseteq \sum_{k}E_k^\dagger E_k\sqsubseteq\Id$,
where $\sqsubseteq$ is the Loewner order defined by $A\sqsubseteq B$ iff $B-A$ is positive semi-definite.
After applying $\calE$ to a mixed state $\rho$,
one will obtain the state $\calE\parens*{\rho}=\sum_{k}E_k\rho E_k^\dagger$.
It is easy to see the Kraus representation of a quantum gate $U$ is simply $\braces*{U}$,
and of a measurement $\braces*{M_m}_m$ is simply $\braces*{M_m}_m$.
Note that these two examples are actually trace-preserving quantum operations,
i.e., $\tr\parens*{\calE\parens*{\rho}}=\tr\parens*{\rho}$.
In general, quantum operations need not be trace-preserving.
For example, from a measurement $\braces*{M_m}_m$,
we can select a partial measurement $\braces*{M_{m_0},M_{m_1}}$,
corresponding to two certain outcomes $m_0$ and $m_1$.
Then, $\braces*{M_{m_0},M_{m_1}}$ is also a quantum operation.
We use $\QO\parens*{\calH}$ to denote the set of quantum operations on $\calH$.

We adopt the following notations for convenience.
In context without ambiguity,
for a quantum operation $\calE\in \QO\parens*{\calH}$,
we use the same notation $\calE$ to represent its extension $\calE\otimes \Id\in \QO\parens*{\calH\otimes \calH'}$
to a larger Hilbert space $\calH\otimes\calH'$,
where $\Id\in \QO\parens*{\calH'}$ is the identity operator.
In this case, $\calE\in \QO\parens*{\calH}$ only indicates that $\calE$ acts non-trivially on $\calH$.
As an example of this convention,
consider two quantum operations $\calE\in \QO\parens*{\calH_1\otimes\calH_2}$
and $\calF\in \QO\parens*{\calH_1\otimes\calH_3}$,
then $\calE+\calF=\calE\otimes \Id_3+\calF\otimes\Id_2\in \QO\parens*{\calH_1\otimes\calH_2\otimes\calH_3}$,
where $\Id_2\in \QO\parens*{\calH_2}$ and $\Id_3\in \QO\parens*{\calH_3}$
are identity operators on $\calH_2$ and $\calH_3$, respectively.

\subsection{Probability Theory}
\label{sub:probability_theory}

In this section we briefly introduce the basic notions and notations in probability theory.
For a more detailed introduction, the readers are referred to the textbook~\cite{Klenke13}.
We need the measure theory for probability, 
because the state space of events considered in this paper is as large as $\R$.
In particular, we need to deal with non-terminating distributed quantum system,
in which an observable event can be about countably many classical outcomes from quantum measurements.
For example, suppose that a quantum system repeatedly performs some quantum gates followed by measurements on a qubit,
then we can talk about the probability of the event that the measurement outcome is always $1$.
More generally, we can talk about the probabilities of any safety (i.e., something bad never happens)
or liveness (i.e., something good eventually happens) properties of the system,
which are of great concern in the design and analysis of systems.

In the following we model the space of elementary events,
the set of observable events, and finally a probability measure.
Let $\Omega$ be a set of elementary events.
Consider a class of subsets $\calA\subset 2^{\Omega}$.
We first define semiring and $\sigma$-algebra.

\begin{definition}[Semiring]
	\label{def:semi-ring}
	A class of sets $\calA\subset 2^{\Omega}$ is a semiring if $\calA$ satisfies the following properties:
	\begin{enumerate}
		\item 
			$\emptyset\in \calA$.
		\item
			For any $X,Y\in \calA$, 
			the difference $X-Y$ is a finite disjoint union of sets in $\calA$.
		\item
			For any $X,Y\in \calA$,
			$X\cap Y\in \calA$.
	\end{enumerate}
\end{definition}

\begin{definition}[$\sigma$-algebra]
	A class of sets $\calA\subset 2^{\Omega}$ is a $\sigma$-algebra if $\calA$ satisfies the following properties:
	\begin{enumerate}
		\item 
			$\Omega\in \calA$.
		\item
			For any $X\in \calA$, $X^{\complement}\in \calA$.
		\item
			For any $\braces*{X_k}_{k\in \N}$ with $X_k\in \calA$,
			the countable union $\bigcup_{k\in \N} X_k\in\calA$.
	\end{enumerate}
\end{definition}

A $\sigma$-algebra satisfies the natural properties of a set of observable events,
on which one can define a probability measure consistently.
Consider the following properties of set functions.

\begin{definition}
	Let $\calA\subset 2^{\Omega}$
	and $\mu:\calA\rightarrow \bracks*{0,1}$ be a set function.
	We say that
	\begin{itemize}
		\item 
			$\mu$ is additive if 
			$\mu\parens*{X}=\sum_{k=1}^K \mu\parens*{X_k}$
			for any $X\in \calA$ and finitely many $X_1, X_2,\ldots,X_K\in \calA$ with $X=\biguplus_{k=1}^K X_k$ and $K\in \N$.
		\item 
			$\mu$ is $\sigma$-additive if 
			$\mu\parens*{X}=\sum_{k\in \N} \mu\parens*{X_k}$
			for any $X\in \calA$ and countably many $X_1, X_2,\ldots\in \calA$ with $X=\biguplus_{k\in \N} X_k$.
		\item
			$\mu$ is $\sigma$-subadditive if
			$\mu\parens*{X}\leq \sum_{k\in \N} \mu\parens*{X_k}$
			for any $X\in \calA$ and countably many $X_1,X_2,\ldots \in \calA$
			with $X\subset \bigcup_{k\in \N} X_k$.
	\end{itemize}
\end{definition}

Then, a probability measure is defined as follows.

\begin{definition}[Probability measure]
	\label{def:pr-meas}
	Let $\calA\subseteq 2^{\Omega}$ be a $\sigma$-algebra,
	and $\mu:\calA\rightarrow \bracks*{0,1}$ be a set function.
	We say $\mu$ is a probability measure on $\calA$ if
	$\mu$ is $\sigma$-additive, $\mu\parens*{\emptyset}=0$ and $\mu\parens*{\Omega}=1$.
\end{definition}

An important lemma in the probability theory is the measure extension lemma,
enabling one to extend a properly defined set function on a semiring $\calA$
to the $\sigma$-algebra $\sigma\parens*{\calA}$ generated by $\calA$, as follows.

\begin{lemma}[Carath{\'e}odory's measure extension, special case of Theorem 1.53 in~\cite{Klenke13}]
	\label{lmm:cara ext}
	Let $\calA\subset 2^\Omega$ be a semiring with $\Omega\in \calA$,
	and $\sigma\parens*{\calA}$ be the $\sigma$-algebra generated by $\calA$.
	Let $\mu:\calA\rightarrow \bracks*{0,1}$ be an additive, $\sigma$-subadditive function on $\calA$
	with $\mu(\emptyset)=0$ and $\mu\parens*{\Omega}=1$,
	then $\mu$ has a unique extension to a probability measure $\tilde{\mu}:\sigma(\calA)\rightarrow \bracks*{0,1}$, and
	\begin{equation*}
		\tilde{\mu}(X)=\inf\braces*{\sum_{k\in \N} \mu(X_k):X\subset \bigcup_{k\in \N} X_k\wedge \forall k, X_k\in \calA}.
	\end{equation*}
\end{lemma}

\section{A Model of Distributed Quantum System}
\label{sec:a_model_of_distributed_quantum_system}

In this section we present a formal model of non-atomic distributed quantum system.
An informal introduction of this model was already presented in \Cref{sub:non-atomic-dis-q-sys}.
We start with some notation conventions.
Readers are referred to \Cref{sub:quantum_computing} for basic notions and notations in quantum computing.
Let $\Qbit$ be a countable set of qubits,
which are controllable objects in quantum computing.
A qubit $x\in \Qbit$ has a Hilbert space $\calH_{x}=\Co^{2}$.
A set of qubits forms a quantum register.
A quantum register $q\subset \Qbit$ has a Hilbert space $\calH_{q}=\bigotimes_{x\in q}\calH_x$.
Let $\Qpar$ be a set of quantum particles,
which are introduced by performing actual quantum operations on qubits
(in particular, the quantum measurements),
and contain degrees of freedom not fully controllable.
A quantum particle $x\in \Qpar$ has a Hilbert space $\calH_x$
of either discrete or continuous dimension depending on $x$.
A set of quantum particles forms an environment.
A quantum environment $e\subset \Qpar$ has a Hilbert space $\calH_{e}=\bigotimes_{x\in e}\calH_x$.
Let $\mathbb{I}\parens*{\R_{\geq 0}}:=\braces*{[x,y]:x\leq y\in \R_{\geq 0}}$
be the set of closed intervals on $\R_{\geq 0}$.
For a Hilbert space $\calH$, let $\calD\parens*{\calH}$ be the set of density operators on $\calH$ and
$\QO\parens*{\calH}$ be the set of quantum operations on $\calH$.

\subsection{Action}
\label{sub:actions}

In quantum computing, an action is a set of bounded space-time events for performing a quantum operation.\footnote{
			How about actions that perform classical operations?
			It is well known that classical computation can be simulated by quantum computation with little overhead.
			For simplicity, we only consider actions that perform quantum operations.
		}
It models a physical implementation or execution of an logical operation,
and therefore includes both the information of the high-level operation and the low-level implementation.
We are concerned about the following properties of an action $a$:
\begin{enumerate}
	\item 
		$T[a]\in \mathbb{I}(\R_{\geq 0})$ represents the time interval that $a$ spans.
	\item
		$q[a]\subset \Qbit$ represents the quantum register on which $a$ is supposed to perform.
	\item
		$\calE[a]\in \QO\parens*{\calH_{q[a]}}$ represents the (logical) quantum operation that $a$ is supposed to perform,
		which is either a unitary quantum gate or a (partial) measurement.
	\item
		$e[a]\subset \Qpar$ represents the quantum environment introduced by $a$.
		If $\calE[a]$ is a unitary,
		then $e[a]=\emptyset$,
		because in this case $q[a]$ is effectively isolated from $e[a]$.
		If $\calE[a]$ is a (partial) measurement,
		then $e[a]$ includes the (possibly infinitely dimensional)
		quantum degrees of freedom in the measurement device.
\end{enumerate}

In the above properties, 
$q[a]$ and $\calE[a]$ together specify the logical operation that $a$ is supposed to perform.
Meanwhile,
$T[a]$ and $e[a]$ are physical properties,
which depends on how $a$ is actually implemented.

\begin{remark}
	\label{rmk:quantum environments}
	Why are quantum environments of concern?
	Recall that to resolve the challenge caused by quantum entanglement (see \Cref{sec:our_contributions}),
	we need to describe the real-time state of a distributed quantum system at any time.
	So we have to include all effective quantum degrees of freedom,
	including those uncontrollable ones in the quantum environments.
\end{remark}

We adopt the following terminologies and notations:
\begin{itemize}
	\item 
		Let $\calH_{a}:=\calH_{q[a]}\otimes \calH_{e[a]}$.
		Then, the space-time region occupied by $a$ is $\calH_{a}\times T[a]$.
	\item
		We use $\Act$ to denote the set of actions.
	\item
		For any countable set $A\subset \Act$ and a time region $X\subset \R_{\geq 0}$,
		let $A\restriction_X:=\braces*{a\in A:T[a]\cap X\neq \emptyset}$.
	\item
		For any countable set $A\subseteq \Act$, let $q[A]:=\bigcup_{a\in A}q[a]$.
		We use the same convention for $e[\cdot]$
		and let $\calH_A:=\calH_{q[A]}\otimes \calH_{e[A]}$.
\end{itemize}

\subsection{Quantum Process}
\label{sub:quantum_process}

A quantum process is a collection of countably many actions with a tree structure.
Operationally, a quantum process repeatedly takes actions as the time progresses.
We can connect two successive actions taken by the quantum process by a relation $\rightarrow$.
If the quantum process does not terminate,
then countably many actions are involved.
The tree structure is a consequence of the probabilistic branchings created by quantum measurements:
a quantum measurement produces finitely many probabilistic branches,
each of which corresponds to an action performing a partial measurement.
The formal definition is as follows.

\begin{definition}[Quantum process]
	\label{def:quantum-proc}
	A quantum process (abbreviated as process) is a tuple $(A, \rightarrow)$ with countable $A\subseteq \Act$
	and a relation $\rightarrow$ on $A$ such that:
	\begin{defenum}
		\item
			\label{den:proc-rt-tree}
			(Rooted tree)
			$(A, \rightarrow)$ is a rooted tree with vertices in $A$ and edges in $\rightarrow$.
		\item 
			\label{den:proc-seq}
			(Sequentiality)
			$\forall a,b\in A, a\rightarrow b\Rightarrow T[a]<T[b]$.\footnote{
				In this paper, we denote $X<Y$ if $\forall t_x\in X,t_y\in Y, t_x<t_y$ for $X,Y\subseteq \R_{\geq 0}$.
			}
		\item
			\label{den:proc-br}
			(Branching)
			$\forall a,b,c\in A, a\rightarrow b\wedge a\rightarrow c\Rightarrow q[b]=q[c]$.
			Moreover, $\forall a\in A,\sum_{b:a\rightarrow b} \calE[b]$ is a quantum operation.
		\item
			\label{den:proc-fin}
			(Finitely many actions in finite time)
			$\forall t\in \R_{\geq 0}, A\restriction_{[0,t]}$ is finite.
	\end{defenum}
\end{definition}

The above definition of quantum process is conceptually the simplest we can think of:
it does not presume any structured computational model.

\begin{remark}
	We explain the conditions in \Cref{def:quantum-proc} as follows:
	\Cref{den:proc-rt-tree} says that the next possible action taken by a quantum process
	depends on all previous actions.
	\Cref{den:proc-seq} says that a quantum process is sequential.
	\Cref{den:proc-br} says that all possible next actions taken by a quantum process
	are consistently from the same quantum operation;
	in particular, the partial measurement actions in a branching of the tree
	should correspond to different classical outcomes from the same quantum measurement.
	\Cref{den:proc-fin} says that the number of actions that begins before any time $t$ is finite,
	which is similar to Axiom A5 in~\cite{Lamport86c}.
\end{remark}

We adopt the following terminologies and notations:
\begin{itemize}
	\item 
		In context without ambiguity, we simply use $A$ to denote $(A,\rightarrow)$.
	\item
		We use $\Proc$ to denote the set of quantum processes.
	\item
		We say $A$ is \textit{trace-preserving}, if in \Cref{den:proc-br} the quantum operation is trace-preserving.
	\item
		We say $A$ is \textit{aligned}, if in addition to \Cref{den:proc-br}, we have
		$a\rightarrow b\wedge a\rightarrow c\Rightarrow \min T\bracks*{b}=\min T\bracks*{c}\wedge e\bracks*{b}=e\bracks*{c}$.
	\item
		We use $\rt(A)$ to denote the root of the tree $A$.
	\item
		Let $\rightarrow^k$ be the $k^{\textup{th}}$ composition of $\rightarrow$.
		Let $\rightarrow^*$ and $\rightarrow^+$ be the Kleene star and Kleene plus of $\rightarrow$;
		i.e., $a\rightarrow^* b$ if $\exists k\geq 0, a\rightarrow^k b$;
		and $a\rightarrow^+ b$ if $\exists k>0, a\rightarrow^k b$.
\end{itemize}

We further define the concept of partial quantum process, 
which can be thought of as a restriction of the original process,
conditioned on knowing the first finitely many actions.

\begin{definition}[Partial quantum process]
	\label{def:partial-proc}
	For $(A,\rightarrow)\in \Proc$,
	we say that $(B,\rightarrow)\in \Proc$ is a partial (quantum) process of $A$, 
	if $\exists b\in B$ such that $B$ consists of a rooted path to $b$ and the sub-tree rooted at $b$;
	i.e.,
	$B=\braces*{a\in A:\rt(A)\rightarrow^*a \rightarrow^* b}\cup \braces*{a\in A:b\rightarrow^+ a}$.
	In this case, we denote $B=A/b$.
\end{definition}

Intuitively, the rooted path to $b$
selects a prefix of probabilistic branches of concern.
From an operational view, each time a quantum process runs into a branching,
it actually goes into one of the branches;
and when talking about $A/b$, 
we know the first finitely many actions up to $b$ are taken.
Note that it is possible that $B=A/a=A/b$ for $a\neq b$.
We adopt the following terminologies and notations:
\begin{itemize}
	\item 
		For convenience, we additionally define $A$ to be a partial process of itself.
		We use $A/*$ to represent the set of all partial processes of $A$.
	\item
		For a quantum process $A\in \Proc$ and a subset $B\subseteq A$,
		we say that $B$ \textit{has no branching} if $(B,\rightarrow^*)$ is a totally ordered set.
		It is easy to see $\forall a\in A$, 
		$\parens*{A/a}\restriction_{\bracks*{0,\max T\bracks*{a}}}$ has no branching.
		It is also obvious that any $C\subseteq B$ has no branching if $B$ has no branching.
	\item
		For a trace-preserving quantum process $A\in \Proc$,
		we say that a partial process $B\in A/*$ is \textit{trace-preserving after time} $t$,
		if $B=A/a$ for some $a$ such that 
		$t\geq \max T\bracks*{a}$ or for any $b\in B\restriction_{[t,\max T\bracks*{a}]}$, $\calE\bracks*{b}$ is trace-preserving.
		Intuitively, after time $t$, $B$ includes all future branches.
		A special case is that $A$ itself is trace-preserving after time $0$.
\end{itemize}

\subsection{Distributed Quantum System}
\label{sub:distributed_quantum_system}

A (non-atomic) distributed quantum system is a collection of finitely many quantum processes
with non-overlapping quantum environments, recursively defined as follows.

\begin{definition}[Distributed quantum system]
	A (non-atomic) distributed quantum system (abbreviated as system) is defined by the following rules:
	\begin{enumerate}
		\item 
			Any quantum process $A$ is a distributed quantum system.
		\item
			For two distributed quantum systems $A$ and $B$ with $e[A]\cap e[B]=\emptyset$,
			their parallel composition $C$, denoted by $C=A\parallel B$,
			is also a distributed quantum system.
			We identify $A\parallel B$ and $B\parallel A$ as the same system,
			and recursively define $e[C]=e[A]\cup e[B]$.
	\end{enumerate}
\end{definition}

\begin{remark}
	The condition $e[A]\cap e[B]=\emptyset$ says the quantum environments of the two systems $A$ and $B$ are effectively separated,
	and the only interaction between them is 
	via manipulating the shared quantum register $q[A]\cap q[B]$.\footnote{
		In distributed systems, interprocess communications typically include message passing.
		As is observed in~\cite{Lamport86a,Lamport86c},
		message passing actually can be modeled by a shared memory.
		From a physical perspective, 
		message passing is also performing quantum operations on the transmission media.
		For simplicity, we do not consider message passing.
	}
	In particular, when $A$ is a quantum process, 
	this condition implies that the measurement devices used by $A$ are effectively separated from
	those by $B$.
	Similar assumption also appears in the classical paper~\cite{Lamport90b},
	where two operations from different processes have disjoint sets of private variables.
	An implication of this condition is $A\cap B=\emptyset$.
\end{remark}

We adopt the following notations and naturally extend some definitions for quantum processes to distributed quantum systems:
\begin{itemize}
	\item 
		We use $\Sys$ to denote the set of distributed quantum systems.
		For convenience, we include $\emptyset \in \Sys$,
		and define $C\parallel \emptyset:=C$ for any $C\in \Sys$.
	\item
		For an action $a\in \Act$ and a system $C=A\parallel B$,
		we define $a\in C$ if $a\in A$ or $a\in B$;
		in other words, as sets $C=A\uplus B$.
	\item
		We say $C=A\parallel B$ is trace-preserving (resp.\ aligned), if $A$ and $B$ are both trace-preserving (resp.\ aligned).
	\item
		We say $C$ is a partial quantum system of another $C'$, denoted by $C\in C'/*$,
		if $C=A\parallel B$, $C'=A'\parallel B'$ and $A\in A'/*$ and $B\in B'/*$.
	\item
		For $D\subseteq C$,
		we say $D$ has no branching, if for any $A\in \Proc, B\in \Sys$ with $C=A\parallel B$,
		$A\cap D$ has no branching.
	\item
		We say $C=A\parallel B$ is trace-preserving after time $t$,
		if $A$ and $B$ are trace-preserving after time $t$.
\end{itemize}

\begin{figure}
	\centering
	\begin{tikzpicture}
		\node (0) at (-6.25, -5.75) {};
		\node (1) at (-5.25, -5.75) {};
		\node (2) at (-4.5, -5.5) {};
		\node (3) at (-4, -6) {};
		\node (4) at (-3.5, -5.5) {};
		\node (5) at (-2.5, -5) {};
		\node (6) at (-2.25, -5.5) {};
		\node (7) at (-0.75, -5) {};
		\node (8) at (0.75, -4.75) {};
		\node (9) at (0.25, -5.25) {};
		\node (10) at (-2.25, -6) {};
		\node (11) at (-1.5, -5.75) {};
		\node (12) at (-0.75, -6) {};
		\node (13) at (-1.75, -6.5) {};
		\node (14) at (0.5, -6.5) {};
		\node (15) at (-0.5, -5.75) {};
		\node (16) at (-1.75, -5) {};
		\node (17) at (0.5, -5.75) {};
		\node (18) at (1.5, -5.75) {};
		\node (19) at (2.5, -5.5) {};
		\node (20) at (3, -6.25) {};
		\node (21) at (-5.75, -7.5) {};
		\node (22) at (-4.75, -7.5) {};
		\node (23) at (-3.75, -7) {};
		\node (24) at (-2.25, -7.5) {};
		\node (25) at (-4.25, -8) {};
		\node (26) at (-3.5, -8) {};
		\node (27) at (-2.75, -8) {};
		\node (28) at (-2.5, -7) {};
		\node (29) at (0, -7) {};
		\node (30) at (1.5, -6.75) {};
		\node (31) at (1.5, -7.25) {};
		\node (32) at (-1.75, -8) {};
		\node (33) at (-2.25, -8.5) {};
		\node (34) at (-0.75, -7.75) {};
		\node (35) at (-0.5, -8) {};
		\node (36) at (0.75, -7.75) {};
		\node (37) at (2.5, -7.75) {};
		\node (38) at (-1.25, -8.5) {};
		\node (39) at (1.25, -8.25) {};
		\node (40) at (0.75, -8.75) {};
		\node (41) at (-6.5, -4.5) {};
		\node (42) at (-6.5, -9) {};
		\node (43) at (4.75, -9) {};
		\node (44) at (-7, -5.75) {$A$};
		\node (45) at (-7, -7.5) {$B$};
		\node (46) at (-1.9, -5.67) {$a$};
		\node (47) at (-3.9, -8.2) {$b$};
		\node (48) at (-4.35, -4.5) {};
		\node (49) at (-4.35, -9) {};
		\node (50) at (-0.9, -4.5) {};
		\node (51) at (-0.9, -9) {};
		\node (52) at (-4.35, -9.25) {$t_1$};
		\node (53) at (-0.9, -9.25) {$t_2$};
		\node (55) at (4, -5.75) {$\ldots$};
		\node (56) at (4, -7.75) {$\ldots$};
		\draw [ultra thick] (0.center) to (1.center);
		\draw [thick] (1.center) to (2.center);
		\draw [ultra thick] (1.center) to (3.center);
		\draw [thick] (2.center) to (4.center);
		\draw [thick] (4.center) to (5.center);
		\draw [thick] (4.center) to (6.center);
		\draw [thick] (7.center) to (8.center);
		\draw [thick] (7.center) to (9.center);
		\draw [ultra thick] (3.center) to (10.center);
		\draw [ultra thick] (10.center) to (11.center);
		\draw [thick] (10.center) to (12.center);
		\draw [thick] (10.center) to (13.center);
		\draw [thick] (13.center) to (14.center);
		\draw [dotted,ultra thick] (11.center) to (15.center);
		\draw [dotted,thick] (5.center) to (16.center);
		\draw [thick] (16.center) to (7.center);
		\draw [ultra thick] (15.center) to (17.center);
		\draw [dotted,ultra thick] (17.center) to (18.center);
		\draw [ultra thick] (18.center) to (19.center);
		\draw [ultra thick] (18.center) to (20.center);
		\draw [ultra thick] (21.center) to (22.center);
		\draw [thick] (22.center) to (23.center);
		\draw [thick] (22.center) to (24.center);
		\draw [ultra thick] (22.center) to (25.center);
		\draw [ultra thick] (25.center) to (26.center);
		\draw [dotted,ultra thick] (26.center) to (27.center);
		\draw [dotted,thick] (23.center) to (28.center);
		\draw [thick] (28.center) to (29.center);
		\draw [thick] (29.center) to (30.center);
		\draw [thick] (29.center) to (31.center);
		\draw [ultra thick] (27.center) to (32.center);
		\draw [ultra thick] (27.center) to (33.center);
		\draw [ultra thick] (32.center) to (34.center);
		\draw [ultra thick] (32.center) to (35.center);
		\draw [dotted,ultra thick] (34.center) to (36.center);
		\draw [ultra thick] (36.center) to (37.center);
		\draw [ultra thick] (33.center) to (38.center);
		\draw [ultra thick] (38.center) to (39.center);
		\draw [ultra thick] (38.center) to (40.center);
		\draw (41.center) to (42.center);
		\draw [->] (42.center) to (43.center);
		\draw [dashed] (48.center) to (49.center);
		\draw [dashed] (50.center) to (51.center);
	\end{tikzpicture}
	\caption{
		An example of distributed quantum system $A\parallel B$ of two processes $A$ and $B$:
		Each process has a tree structure. Every segment corresponds to an action,
		whose projection onto the time axis is the time interval of this action.
		Every dotted segment represents that there is no action during the time interval.
		Every branching corresponds to a set of measurement actions.
		Two partial processes $A/a$ and $B/b$ are highlighted in bold.
		For illustration, $\parens*{A/a}\restriction_{[0,t_2]}$ has no branching,
		and $B/b$ is trace-preserving after time $t_1$.
	}
	\label{fig:sys-eg}
\end{figure}
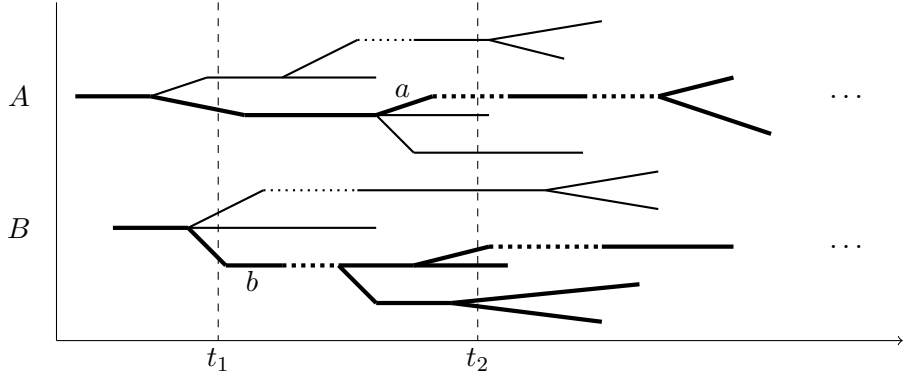

In \Cref{fig:sys-eg}, an example of a distributed quantum system of two processes is presented.

In a system, we can identify those local actions as follows.

\begin{definition}[Local actions]
	\label{def:local-act}
	Given a distributed quantum system $C=A\parallel B$
	for some process $A\in \Proc$ and system $B\in \Sys$,
	an action $a\in A$ is said to be local,
	if $\forall b\in B, q[a]\cap q[b]=\emptyset\vee T[a]\cap T[b]=\emptyset$.
	Intuitively, $a$ is space-time separated from
	actions in other processes.
\end{definition}

\subsection{Dijkstra-Lamport Condition}
\label{sub:dijkstra_lamport_condition}

Before proceeding,
let us briefly revisit the Dijkstra-Lamport condition.
How to correspond what an action is supposed to do to what it actually performs?
A natural condition is: ``\textit{any local action is performed correctly.}''
This condition seems to be firstly mentioned by Dijkstra implicitly in \cite{EWD123}
and \cite{EWD198}.
In particular, in \cite{EWD123} he pointed out ``the nature of a single sequential process,
performing its sequence of actions autonomously,
i.e., independent of its surroundings'';
and in \cite{EWD198} he observed ``the effect of actions is only defined by and describable in a projection of a subspace\ldots each process is related to its own subspace.''
Later, Lamport first explicitly used this condition in reasoning about non-atomic read and write actions:
a read that does not overlap any write must obtain the correct value~\cite{Lamport79,Lamport86a}.
We therefore call it Dijkstra-Lamport condition.
Though being simple,
it is an important condition for us to define the system dynamics of distributed quantum systems.

\subsection{System Dynamics}
\label{sub:system_dynamics}

In this section we describe the system dynamics of a trace-preserving distributed quantum system $S$,
which characterises the real-time effect (semantics) of the system;
that is, for any partial system $C\in S/*$ and time $t\in \R_{\geq 0}$,
we define the quantum operation that maps the initial state at time $0$ to the state at time $t$,
evolved according to $C$.\footnote{
	Why is the map a quantum operation?
	The subtlety is that we take into account all (effective) quantum degrees of freedom,
	including those uncontrollable ones in the quantum environments (in particular, in the quantum measurement devices).
	In this case, the state at a time contains all quantum information of concern,
	and the evolution of state can be represented by a quantum operation 
	(it can be actually a unitary; for example, see~\cite{Schlosshauer05}).
}

For any action $a\in C$, suppose that $T[a]=[x,y]$.
Note that $e[a]$ can be effectively decoupled with $q[a]$ before time $x$ and after time $y$.
W.l.o.g., we can assume that before $x$ and after $y$,
the state of $e[a]$ is a special pure state $\ket{\bot}$,
representing $e[a]$ is not being used.
For a quantum operation $\calF\in \QO\parens*{\calH_{C}}$
and initial quantum state $\rho\in \calD\parens*{\calH_{q[C]}}$,
we denote $\calF\parens*{\rho}=\calF\parens*{\rho\otimes \ket{\bot}\!\bra{\bot}}$
for $\ket{\bot}\in \calH_{e[C]}$.

\begin{definition}[System dynamics]
	\label{def:sys-dyn}
	Given a trace-preserving distributed quantum system $S\in \Sys$,
	the system dynamics of $S$ is a function $\Bracks*{\cdot}\parens*{\cdot}$
	such that for any partial system $C\in S/*$ and time $t\in \R_{\geq 0}$,
	$\Bracks*{C}\parens*{t}\in \QO\parens*{\calH_{C}}$ is a quantum operation satisfying
	the following conditions:
	\begin{defenum}
		\item
			(Initial condition) 
			\label{den:sys-dyn-init-cnd}
			$\Bracks*{C}(0)=\Id$.
		\item
			(Branching)
			\label{den:sys-dyn-br}
			If $C=A/a\parallel B$ for some $A\in \Proc, B\in \Sys, a\in A$
			and $\max_{b:a\rightarrow b}\max T\bracks*{b}\leq t$,\footnote{
				We use the convention that on $\R_{\geq 0}$,
				$\min \emptyset := +\infty$ and $\max \emptyset :=0$.
			}
			then
			\begin{equation}
				\Bracks*{C}(t)=\sum_{b:a\rightarrow b}\Bracks*{A/b\parallel B}(t).\label{eq:sys-dyn-br}
			\end{equation}
		\item
			(Evolution)
			\label{den:sys-dyn-evo}
			For any time interval $I=[x,y]$, 
			if $C\restriction_{[0,x)}$ has no branching, then
			\begin{equation}
				\Bracks*{C}(y)=\calF\circ\Bracks*{C}(x-)\footnote{
					For any function $f(x)$, we use $f(x-):=\lim_{x'\rightarrow x-}f(x')$
					to denote the left limit of $f$ at $x$.
				}
				\label{eq:sys-dyn-evo}
			\end{equation}
			for some quantum operation $\calF\in \QO\parens*{\calH_{C}}$ uniquely determined by $C\restriction_I$.

			Moreover, suppose that $C=A\parallel B$ for some $A,B\in \Sys$.
			If $q\bracks*{A\restriction_I}\cap q\bracks*{B\restriction_I}=\emptyset$,
			then $\calF=\calF_A\otimes \calF_B$ for some quantum operations $\calF_A$ and $\calF_B$
			uniquely determined by $A\restriction_I$ and $B\restriction_I$, respectively.
			In particular, $\calF_A=\calE\bracks*{a}$ if $A\restriction_I=\braces*{a}$ for some local $a$ and $I=T\bracks*{a}$;
			and $\calF_A=\Id$ if $A\restriction_I=\emptyset$.

		\item
			(Trace)
			\label{den:sys-dyn-tr}
			For any initial quantum state $\rho\in \calD\parens*{\calH_{q[C]}}$,
			$\tr\parens*{\Bracks*{C}(t)(\rho)}$ is non-increasing with respect to $t$.
			Moreover, if $C$ is trace-preserving after time $t_0$,
			then for any $t>t_0$,
			$\tr\parens*{\Bracks*{C}\parens*{t_0}\parens*{\rho}}=\tr\parens*{\Bracks*{C}\parens*{t}\parens*{\rho}}$.
	\end{defenum}
\end{definition}

\begin{remark}
	We explain the conditions in \Cref{def:sys-dyn} as follows:
	\Cref{den:sys-dyn-init-cnd} says at time $0$ the system does nothing.
	\Cref{den:sys-dyn-br} says at any time after a branching in a process,
	the state evolved according to $C$ (containing all branches)
	is a mix over states, each evolved according to a partial system containing exactly one branch.

	\Cref{den:sys-dyn-evo} is the most important condition.
	The first part says the evolution of the state according to $C$ in a time interval $I$
	is uniquely determined by all actions in $C$ that overlap with $I$, 
	provided that the state 
	at the beginning of $I$ (i.e., time $x$) corresponds to a single branch.
	The second part further says 
	if $C$ is a parallel composition of two systems $A$ and $B$ with no interaction in $I$,
	then the two systems evolve separately.
	In particular,
	it contains a formalisation of the Dijkstra-Lamport condition (see \Cref{sub:dijkstra_lamport_condition}):
	if a process in the first system $A$ takes a local action $a$ with time interval $I$,
	then the action correctly performs what it is supposed to do ($\calE[a]$).
	It is also natural that if a process has no action in $I$,
	its does nothing.

	\Cref{den:sys-dyn-tr} says the trace of the state evolved according to $C$
	is non-increasing, and becomes constant after time $t_0$,
	if $C$ is trace-preserving after time $t_0$.
\end{remark}

The conditions in \Cref{def:sys-dyn} are the simplest we can think of.
There are many possible functions $\Bracks*{\cdot}\parens*{\cdot}$
satisfying \Cref{def:sys-dyn}. 
In the remainder of this paper, when we talk about the system dynamics,
we arbitrarily pick one that satisfies \Cref{def:sys-dyn}.

\begin{remark}
	\label{rmk:sys-dyn-evo-open-int}
	\Cref{den:sys-dyn-evo} actually also implies the case $I=[x,y)$
	by taking a limit of $y$,
	in which \Cref{eq:sys-dyn-evo} still holds except that $y$ is replaced by $y-$.
	Similar statements hold for the cases $I=(x,y]$ and $I=(x,y)$
	by replacing $C\restriction_{[0,x)}$ with $C\restriction_{[0,x]}$ and $x-$ with $x$.
\end{remark}

\begin{remark}
	\label{rmk:circ-measurement-pr}
	How is the quantum measurement problem circumvented?
	Note that in \Cref{den:sys-dyn-br},
	\Cref{eq:sys-dyn-br} only holds at time $t$ 
	after all measurement actions (in a branching) finish.
	It does not depend on any details of how they are implemented,
	and we do not specify when the branching occurs.
\end{remark}

\begin{remark}
	Since any partial quantum system $C$ is trace-preserving after some sufficiently large time $t_C$,
	according to \Cref{den:sys-dyn-tr},
	we have $\tr\parens*{\Bracks*{C}\parens*{t_C}\parens*{\rho}}=\lim_{t\to\infty}\tr\parens*{\Bracks*{C}\parens*{t}\parens*{\rho}}$.
\end{remark}


\subsection{Observable Dynamics}
\label{sub:observable_dynamics}

In this section we define the observable dynamics of a distributed quantum system $S$,
which based on the system dynamics,
characterises the probabilities of all classically observable events 
in the systems;
that is, for any partial system $C\in S/*$ and any state $\rho$,
we define the probability that $C$ is classically observed, starting with initial state $\rho$.
This can be thought of as the semantics of the system that we can observe.
Readers are referred to \Cref{sub:probability_theory} for basic notions and notations in probability theory.

We first identify the elementary observable events in a system,
each corresponding to a maximal path, defined as follows.

\begin{definition}[Maximal path]
	For a tree $(A,\rightarrow)$,
	a subset $B\subseteq A$ is called a path if 
	$B=\braces*{a_1,a_2,\ldots}$ and $a_j\rightarrow a_{j+1}$ for all $j$.
	A path $B\subseteq A$ is maximal if there does not exist another path $C\subseteq A$
	such that $B\subsetneq C$.
	It is easy to see that every maximal path is rooted.
\end{definition}

For a quantum process $A\in \Proc$,
let $\omega(A)$ denote the set of maximal paths in the tree $A$.
For a distributed quantum system $C=A\parallel B$,
let $\omega\parens*{C}:=\omega\parens*{A}\times \omega\parens*{B}$.\footnote{
	Here, note that the Cartesian product $\times$ implicitly assumes an order between $A$ and $B$.
	In defining $\omega\parens*{C}$, 
	we assume a specific order is chosen and used consistently throughout the paper.
}

For a distributed quantum system $S\in \Sys$,
let $\calS:=\braces*{\omega\parens*{C}:C\in S/*}\cup \braces*{\emptyset}\subseteq \calP\parens*{\omega\parens*{S}}$\footnote{
    In this paper, for a set $X$, we use $\calP\parens*{X}$ to denote the power set of $X$.
}
be the class of sets of maximal paths in partial systems of $S$ (together with an empty set).
The following lemma shows that $\calS$ is a semiring.
For readability, the proof is deferred to \Cref{sub:proof_of_lmm_calS_semiring}.

\begin{lemma}
	\label{lmm:calS-semiring}
	For a distributed quantum system $S$,
	the class $\calS\subset \calP\parens*{\omega\parens*{S}}$ forms a semiring.
\end{lemma}

Let $\sigma\parens*{\calS}\subset \calP\parens*{\omega\parens*{S}}$ be the $\sigma$-algebra generated by $\calS$.
Intuitively, $\sigma\parens*{\calS}$ is the set of classically observable events to which we can assign probabilities.
Then, we can define a probability measure on $\sigma\parens*{\calS}$ by the following lemma, based on the system dynamics.
For readability, the proof is deferred to \Cref{sub:proofs_of_mu_is_pr_meas}.

\begin{lemma}
	\label{lmm:mu-is-pr-meas}
	For a trace-preserving distributed quantum system $S\in \Sys$ and an initial quantum state $\rho\in \calD\parens*{\calH_{q[S]}}$
	with $\tr\parens*{\rho}=1$,
	there exists a unique probability measure $\mu_{\rho\rightarrow S}:\sigma\parens*{\calS}\rightarrow \bracks*{0,1}$
	such that for any partial system $C\in S/*$,
	\begin{equation*}
		\mu_{\rho\rightarrow S}\circ\omega\parens*{C}=\lim_{t\to\infty}\tr\parens*{\Bracks*{C}(t)(\rho)}.
	\end{equation*}
\end{lemma}

From \Cref{lmm:mu-is-pr-meas}, 
we can use the function $\mu_{\cdot \rightarrow S}\parens*{\cdot}$ to define the observable dynamics of $S$.
For convenience, we instead formally define the equivalence between two systems (up to the observable dynamics).
To this end, we introduce the isomorphism between two systems.
Intuitively, two systems are isomorphic if they are supposed to physically implement the same logical system,
which can be thought of as abstracting out the physical properties (i.e., $T[a]$ and $e[a]$ for every action $a$ within the systems).

\begin{definition}[Isomorphism]
	\label{def:isomorphism}
	For two trace-preserving distributed quantum systems $S,S'\in \Sys$,
	a bijection $\gamma:S\rightarrow S'$ is an isomorphism if
	\begin{enumerate}
		\item
			$\gamma$ preserves $\rightarrow$:
			$\forall a,b\in S,a\rightarrow b\Leftrightarrow \gamma\parens*{a}\rightarrow \gamma\parens*{b}$.
		\item
			$\gamma$ preserves $q$ and $\calE$:
			$\forall a\in S, q\bracks*{a}=q\bracks*{\gamma\parens*{a}}\wedge \calE\bracks*{a}=\calE\bracks*{\gamma\parens*{a}}$.
	\end{enumerate}
\end{definition}

Slightly abusing the notation, for an isomorphism $\gamma:S\rightarrow S'$ and any set $X$,
we recursively define $\gamma\parens*{X}:=\braces*{\gamma\parens*{Y}:Y\in X}$.
Then for $X\in \sigma\parens*{\calS}$,
we have $\gamma\parens*{X}\in \sigma\parens*{\calS'}$.
Now we define the equivalence (up to the observable dynamics) as follows.

\begin{definition}[Equivalent systems]
	\label{def:equi-system}
	Two trace-preserving distributed quantum systems $S,S'\in \Sys$,
	are equivalent (up to the observable dynamics),
	denoted by $S\simeq S'$, 
	if there exists an isomorphism $\gamma:S\rightarrow S'$ such that:
	for any state $\rho\in \calD\parens*{\calH_{q[S]}}$, 
	$\mu_{\rho\rightarrow S}=\mu_{\rho\rightarrow S'}\circ\gamma$.
\end{definition}

Equivalent systems are classically indistinguishable.
Note that $q\bracks*{S}=q\bracks*{S'}$ because $\gamma$ is an isomorphism.

\section{Atomicity of Local Actions}
\label{sec:atomicity_of_local_actions}

Based on the model developed in \Cref{sec:a_model_of_distributed_quantum_system},
we are able to prove our main theorem,
which justifies the atomicity of local actions.
We first define the atomicity in our model.

\begin{definition}[Atomicity]
	\label{def:atomicity}
	In a distributed quantum system $S$,
	a set $D\subseteq S$ of actions are said to be atomic,
	if $\forall a,b\in D$ with $a\in A,b\in B$ for some $A,B\in \Sys$ and $S=A\parallel B$,
	either $T[a]<T[b]$ or $T[b]< T[a]$.
\end{definition}

\begin{remark}
	In \Cref{def:atomicity},
	it is worth noting that the sequentiality condition 
	(either $T[a]<T[b]$ or $T[b]< T[a]$) is only imposed on atomic actions,
	and nothing is presumed for non-atomic actions.
	What is the temporal relation between an atomic and a non-atomic action?
	They can still be concurrent (and not sequential), 
	because the former is ``indivisible'' and the latter is ``divisible''.
\end{remark}

Then, we prove the following intermediate theorem,
stating that
local actions can be regarded as if they were instantaneous,\footnote{
	The idea is similar to the shrinking of time intervals in Proposition 1 in \cite{Lamport86c} (see also Proposition 4 in \cite{Lamport85}).
	The difference is that \cite{Lamport85,Lamport86c} only concern the temporal relations in a system;
	while here we are concerned about the observable dynamics,
	based on the real-time states of a system.
	See also \Cref{sub:related_works} for a discussion of related works.
}
up to the observable dynamics of the system.


\begin{theorem}[Local actions regarded instantaneous]
	\label{thm:local-ins}
	Given a trace-preserving distributed quantum system $S\in \Sys$,
	there exist another trace-preserving system $S'\simeq S$ and an isomorphism $\gamma:S\rightarrow S'$ such that:
	\begin{enumerate}
		\item
			For any local action $a\in S$,
			$e\bracks*{\gamma\parens*{a}}=\emptyset$ and
			$T\bracks*{\gamma\parens*{a}}=\braces*{t_a}$ for some $t_a\in \R_{\geq 0}$.
		\item
			For any non-local action $a\in S$, $\gamma(a)=a$.
	\end{enumerate}
\end{theorem}

We give a proof sketch of \Cref{thm:local-ins}.
For readability, the full proof is deferred to \Cref{sub:proof_of_thm_local_ins_atom}.

\begin{proof}[Proof sketch of \Cref{thm:local-ins}]
	The proof consists of two steps.
	\begin{enumerate}
		\item 
			\label{stp:thm-local-ins-pf-sk-1}
			Given a system $S$ and a local action $a\in S$,
			we can change the time interval $T\bracks*{a}$ of $a$ 
			to an instant $\braces*{t_a}$ with $t_a\in T\bracks*{a}$
			and obtain a new system $S'$.
			Then, $S\simeq S'$.

			To prove $S\simeq S'$, suppose that $T\bracks*{a}=[x,y]$.
			For any partial system $C\in S/*$,
			denote the corresponding partial system (via the isomorphism $\gamma$) of $S'$ by $C'$.
			By \Cref{def:equi-system,lmm:mu-is-pr-meas},
			it suffices to prove $\mu_{\rho\to S}\circ\omega\parens*{C}=\mu_{\rho\to S'}\circ\omega\parens*{C'}$
			for any state $\rho$ and partial system $C$.
			Actually $C$ can be ``decomposed'' (see \Cref{lmm:C-exp-to-after-t,sub:decomposition_lemmas_about_partial_systems}),
			majorly due to \Cref{den:sys-dyn-br},
			and consequently the task can be further reduced to assuming $C\restriction_{[0,y]}$ has no branching.
			We can prove the following equalities step by step,
			via repeated uses of \Cref{den:sys-dyn-evo}:
			\begin{enumerate}
				\item 
					$\Bracks*{C}\parens*{x-}=\Bracks*{C'}\parens*{x-}$;
				\item
					$\Bracks*{C}\parens*{y}=\Bracks*{C'}\parens*{y}$; and
				\item
					$\Bracks*{C}\parens*{t}=\Bracks*{C'}\parens*{t}$ for $t>y$.
			\end{enumerate}
			The second equality is the most complicated to prove,
			where we will resort to the Dijkstra-Lamport condition in \Cref{den:sys-dyn-evo}.
			Taking the trace and $t\to +\infty$ in the last equality leads to 
			the conclusion, according to \Cref{lmm:mu-is-pr-meas}.
		\item
			Given a system $S$,
			we can construct a family of systems $\braces*{S_m}_{m\in \N}$,
			such that each $S_{m}$ is obtained from $S_{m-1}$ as in Step~\ref{stp:thm-local-ins-pf-sk-1} above
			(with each time a new local action chosen).
			Then we have $S\simeq S_1\simeq S_2\simeq\ldots$.
			Taking the limit $S'=\lim_{m\to \infty} S_m$,
			we can verify that all local actions in $S'$ are instantaneous
			and $S'\simeq S$.

			To prove $S'\simeq S$,
			for any partial system $C\in S/*$,
			denote the corresponding partial system of $S_m$ by $C_m$,
			and of $S'$ by $C'$.
			From \Cref{den:sys-dyn-tr}, we can choose sufficiently large $t$ 
			such that $C$ is trace-preserving after time $t$,
			and $m$ such that $C_m\restriction_{[0,t]}=C'\restriction_{[0,t]}$.
			Consequently, by \Cref{den:sys-dyn-evo},
			$\tr\parens*{\Bracks*{C_m}\parens*{t}\parens*{\rho}}=\tr\parens*{\Bracks*{C'}\parens*{t}\parens*{\rho}}$,
			which results in $S'\simeq S_m$ from \Cref{den:sys-dyn-tr,lmm:mu-is-pr-meas}.
			The conclusion follows from $S\simeq S_m\simeq S'$.
	\end{enumerate}
\end{proof}

Based on \Cref{thm:local-ins}, we can prove the main theorem,
stating that local actions can be regarded as if they were atomic,
up to the observable dynamics of the system.

\begin{theorem}[Local actions regarded atomic]
	\label{thm:local-atom}
	Given a trace-preserving and aligned distributed quantum system $S\in \Sys$ with $\forall a\in S, \abs*{T\bracks*{a}}>0$,
	there exists another trace-preserving and aligned system $S'\simeq S$ such that
	local actions in $S'$ are atomic.
\end{theorem}

We give a proof sketch of \Cref{thm:local-atom}.
For readability, the full proof is deferred to \Cref{sub:proof_of_thm_local_ins_atom}.

\begin{proof}[Proof sketch of \Cref{thm:local-atom}]
	Given \Cref{thm:local-ins},
	one can replace all local actions in $S$ with instantaneous versions to obtain a system $S'$.
	Since $\abs*{T\bracks*{a}}>0$ for all $a\in S$,
	it is possible to arrange the instants of these local actions such that they never overlap.
	The conclusion immediately follows.
\end{proof}

\begin{remark}
	Note that any physically implementable system naturally satisfies the conditions of $S$ in \Cref{thm:local-atom}.
	If a system has a physical implementation,
	it preserves the probability and is hence trace-preserving.
	Also, partial measurements from the same measurement are consistently performed by the same device,
	so the system is also aligned.
	Finally, in the real world, no actions are performed instantly,
	which exactly means $\forall a\in S, \abs*{T[a]}>0$.
\end{remark}


\section*{Acknowledgements}

Zhicheng Zhang thanks Qisheng Wang, Angsar Manatuly, Alexander Hahn, Daniel Burgarth and Yangfang Wu for helpful discussions.
Zhicheng Zhang was supported by the Sydney Quantum Academy, NSW, Australia.

\newpage

\printbibliography

\newpage

\appendix

\section{Details of Proofs}
\label{sec:details_of_proofs}

In this section we present the details of all proofs in this paper.
The structure is organised as follows:
in \Cref{sub:technical_lemmas_about_partial_processes}
two technical lemmas about partial processes are presented.
In \Cref{sub:decomposition_lemmas_about_partial_systems}
we show that any partial system can be decomposed into ``finer'' ones,
w.r.t.\ a relation $\rightsquigarrow$, induced by the relation $\rightarrow$ in \Cref{def:quantum-proc}.
We also prove three technical lemmas about such decomposition.
In \Cref{sub:proof_of_lmm_calS_semiring} we prove \Cref{lmm:calS-semiring};
i.e., $\calS$ defined in \Cref{sub:observable_dynamics} is a semiring.
In \Cref{sub:technical_lemmas_about_semiring_cals_} we further prove two technical lemmas about the semiring $\calS$.
Then, in \Cref{sub:proofs_of_mu_is_pr_meas}, using the above lemmas,
we prove \Cref{lmm:mu-is-pr-meas}; i.e.,
$\mu_{\rho\rightarrow S}$ defined in \Cref{sub:observable_dynamics} is a probability measure.
Finally, in \Cref{sub:proof_of_thm_local_ins_atom} we prove the main \Cref{thm:local-ins,thm:local-atom}.

\subsection{Technical Lemmas about Partial Processes}
\label{sub:technical_lemmas_about_partial_processes}

The following lemmas about partial quantum processes are useful and easy to see from the tree structure of quantum processes
and \Cref{def:partial-proc}.

\begin{lemma}
	\label{lmm:partial-br}
	Given any partial process $A/a\in A/*$ for some quantum process $A\in \Proc$ and action $a\in A$,
	if $\braces*{b:a\rightarrow b}\neq \emptyset$,
	then we have $\omega\parens*{A/a}=\biguplus_{b:a\rightarrow b}\omega\parens*{A/b}$.
\end{lemma}

\begin{lemma}
	\label{lmm:partial-cap}
	Given any two partial processes $A/a, A/b\in A/*$ for some quantum process $A\in \Proc$ and actions $a,b\in A$,
	if $\omega\parens*{A/a}\cap \omega\parens*{A/b}\neq \emptyset$,
	then $a\rightarrow^* b \vee b\rightarrow^* a$.
\end{lemma}


\subsection{Decomposition Lemmas about Partial Systems}
\label{sub:decomposition_lemmas_about_partial_systems}

In this section we consider decomposing a partial system into 
``finer'' ones. We first define a relation $\rightsquigarrow$,
induced from the relation $\rightarrow$ in \Cref{def:quantum-proc}.

\begin{definition}[Relation $\rightsquigarrow$]
	\label{def:rel-squiar}
	Given a trace-preserving distributed quantum system $S\in \Sys$,
	define a relation $\rightsquigarrow \subseteq \calP\parens*{S/*}\times \calP\parens*{S/*}$
	on the powerset of $S/*$ such that 
	$X\rightsquigarrow Y$ iff 
	$X-Y=\braces*{A/a\parallel B}$ for some trace-preserving $A\in \Proc$, $a\in A$ and $B\in \Sys$;
	and $Y-X=\braces*{A/b\parallel B:a\rightarrow b}$.
\end{definition}

As usual, we use $\rightsquigarrow^*$ and $\rightsquigarrow^+$ to denote the Kleene star and Kleene plus of $\rightsquigarrow$, respectively.
The following lemma shows that if a finite set of partial systems
is decomposed from another w.r.t.\ the relation $\rightsquigarrow$,
then the unions of their corresponding maximal paths are the same.

\begin{lemma}
	\label{lmm:set-partial-sys-dec}
	Given a trace-preserving distributed quantum system $S\in \Sys$,
	and two finite sets $X,Y\subseteq S/*$ of partial systems,
	if $X\rightsquigarrow^* Y$
	and $\forall C,D\in X, \omega\parens*{C}\cap \omega\parens*{D}=\emptyset$,
	then $\biguplus_{C\in X}\omega\parens*{C}=\biguplus_{C\in Y}\omega\parens*{C}$.\footnote{The inverse direction however does not hold.
		Finding a counter example is left to the readers.
	}
\end{lemma}

\begin{proof}
	We prove the lemma by induction on the transitive closure $\rightsquigarrow^*$.
	\begin{enumerate}
		\item 
			$X=Y$. This case is trivial.
		\item
			$X\rightsquigarrow^* Z$ and $Z\rightsquigarrow Y$ for some finite set $Z\subseteq S/*$.
			By the induction hypothesis, $\biguplus_{C\in X}\omega\parens*{C}=\biguplus_{C\in Z}\omega\parens*{C}$.
			According to \Cref{def:rel-squiar},
			$Z-Y=\braces*{A/a\parallel B}$ for some trace-preserving $A\in \Proc$, $a\in A$
			and $B\in \Sys$; and $Y-Z=\braces*{A/b\parallel B:a\rightarrow b}$.
			By \Cref{lmm:partial-br},
			$\omega\parens*{A/a}=\biguplus_{b:a\rightarrow b}\omega\parens*{A/b}$.
			Thus, $\biguplus_{C\in Z}\omega\parens*{C}=\biguplus_{C\in Y}\omega\parens*{C}$.
			The conclusion immediately follows.
	\end{enumerate}
\end{proof}

Recall that for a quantum process $A$,
any partial process $B\in A/*$ consists of a rooted path to some $a\in A$
and the subtree rooted at $a$.
We define a function $\ell\parens*{\cdot}$ such that $\ell\parens*{B}$
is the minimal length of such rooted path.
It can also be naturally generalised to partial systems.

\begin{definition}[Function $\ell$]
	\label{def:func-ell}
	Given a partial processes $B\in A/*$ for some trace-preserving quantum process $A\in \Proc$,
	define $\ell\parens*{B}=\min\braces*{r: B=A/a\wedge\rt\parens*{A}\rightarrow^r a\in A}\in \N$.

	Given a partial systems $C=A\parallel B$,
	define $\ell\parens*{C}=\max \braces*{\ell\parens*{A},\ell\parens*{B}}$.
\end{definition}

Our next lemma shows that two partial systems can be decomposed w.r.t.\ the relation $\rightsquigarrow$
such that after the decomposition,
the intersection of the two sets of partial systems is a unique partial system,
whose corresponding maximal paths are exactly the intersection of those of the original two partial systems.

\begin{lemma}
	\label{lmm:C-D-dec-X-Y}
	Given a trace-preserving distributed quantum system $S\in \Sys$,
	for any two partial systems $C,D\in S/*$ with $\omega\parens*{C}\cap \omega\parens*{D}\neq \emptyset$,
	there exists two corresponding finite sets $X,Y\subseteq S/*$
	such that 
	\begin{enumerate}
		\item 
			$\braces*{C}\rightsquigarrow^* X$ and $\braces*{D}\rightsquigarrow^* Y$;
		\item
			$X\cap Y=\braces*{E}$ for some $E\in S/*$ and $\omega\parens*{E}=\omega\parens*{C}\cap \omega\parens*{D}$.
	\end{enumerate}
\end{lemma}

\begin{proof}
	We show how to construct $X,Y$ step by step.
	In the construction, we require the intermediate $X,Y$ to satisfy the following:
	\begin{align}
		&\braces*{C}\rightsquigarrow^* X\wedge \braces*{D}\rightsquigarrow^* Y,
		\label{eq:C-D-dec-X-Y-first}\\
		&\exists! (E,F)\in X\times Y, \omega\parens*{E}\cap\omega\parens*{F}\neq \emptyset.
		\label{eq:C-D-dec-X-Y-im-req}
	\end{align}

	The construction is as follows.
	\begin{enumerate}
		\item 
			\label{stp:C-D-dec-X-Y-cstr-1}
			Initially, let $X=\braces*{C}$ and $Y=\braces*{D}$.
		\item
			\label{stp:C-D-dec-X-Y-cstr-2}
			Repeat the following procedure.
			First pick the unique pair $(E,F)\in X\times Y$ according to \Cref{eq:C-D-dec-X-Y-im-req}.

			If $E=F$, then by \Cref{eq:C-D-dec-X-Y-first}, \Cref{eq:C-D-dec-X-Y-im-req} and \Cref{lmm:set-partial-sys-dec},
			it is easy to see that the properties in \Cref{lmm:C-D-dec-X-Y} are satisfied and we can terminate.

			If $E\neq F$, we can write $E=A/a\parallel B$ and $F=A/a'\parallel B'$ for some trace-preserving $A\in \Proc$,
			$a\neq a'\in A$ and $B,B'\in \Sys$.
			From \Cref{eq:C-D-dec-X-Y-im-req}, we have $\omega\parens*{A/a}\cap \omega\parens*{A/a'}\neq \emptyset$.
			Moreover, by \Cref{lmm:partial-cap} and $a\neq a'$, either $a\rightarrow^+ a'$
			or $a'\rightarrow^+ a$.
			Suppose w.l.o.g.\ that $a\rightarrow^+ a'$,
			then let $X'=\parens*{X-\braces*{E}}\cup \braces*{A/b\parallel B:a\rightarrow b}$
			and $Y'=Y$,
			where $X',Y'$ stand for the new values of $X,Y$, respectively.
	\end{enumerate}

	Now we show that \Cref{eq:C-D-dec-X-Y-im-req,eq:C-D-dec-X-Y-first}
	hold in the above construction.
	In Step~\ref{stp:C-D-dec-X-Y-cstr-1}, they obviously hold for the initial $X,Y$.
	Let us verify \Cref{eq:C-D-dec-X-Y-im-req,eq:C-D-dec-X-Y-first} for $X',Y'$
	given that they hold for $X,Y$, in Step~\ref{stp:C-D-dec-X-Y-cstr-2}:
	\begin{itemize}
		\item
			\Cref{eq:C-D-dec-X-Y-first}:
			It simply follows from $X\rightsquigarrow X'$ and $Y\rightsquigarrow^* Y'$.
		\item 
			\Cref{eq:C-D-dec-X-Y-im-req}:
			From above, we have $\omega\parens*{E}=\biguplus_{G\in X'-X}\omega\parens*{G}$.
			From \Cref{eq:C-D-dec-X-Y-im-req} for $X,Y$,
			we also have $\forall E'\in X\cap X',F'\in Y,\omega\parens*{E'}\cap \omega\parens*{F'}=\emptyset$ and
			$\forall F'\neq F\in Y, \omega\parens*{E}\cap \omega\parens*{F'}=\emptyset$.
			Now it suffices to show that $\exists! G\in X'-X, \omega\parens*{G}\cap \omega\parens*{F}\neq \emptyset$.
			Since $a\rightarrow^+ a'$, there exists a unique $b'$ such that $a\rightarrow b'\wedge b'\rightarrow^* a'$.
			Let $G=A/b'\parallel B$,
			then $\omega\parens*{A/b'}\cap \omega\parens*{A/a'}\neq \emptyset$.
			From \Cref{eq:C-D-dec-X-Y-im-req} for $X,Y$,
			$\omega\parens*{B}\cap \omega\parens*{B'}\neq \emptyset$.
			Hence, $\omega\parens*{G}\cap \omega\parens*{F}\neq \emptyset$.
			The uniqueness of $G$ follows from that of $b'$.
	\end{itemize}

	It only remains to show that the above construction terminates.
	Note that in Step~\ref{stp:C-D-dec-X-Y-cstr-2},
	if $X'\neq X$, 
	by construction we have $\forall G\in X'-X$, $\ell\parens*{G}\leq \ell\parens*{F}$.
	A similar statement holds for $Y$.
	By a simple induction, it can be seen that 
	$\forall G\in \parens*{X'-X}\cup \parens*{Y'-Y},\ell\parens*{G}\leq \max\braces*{\ell\parens*{C},\ell\parens*{D}}$.
	Moreover, consider the set 
	\begin{equation*}
		Z=\braces*{G\in S/*:\ell\parens*{G}\leq \max\braces*{\ell\parens*{C},\ell\parens*{D}}}.
	\end{equation*}
	Let us focus on the construction of $X$. Each time when $X'\neq X$ in Step~\ref{stp:C-D-dec-X-Y-cstr-2},
	in $X'-X$ we will visit at least one new element in $Z$.
	Since $Z$ is finite, the number of times when $X'\neq X$ should be finite.
	A similar statement holds for $Y$.
	Hence, the above construction terminates.
\end{proof}

Our last lemma in this section shows that if for two finite sets of partial systems,
the unions of their corresponding maximal paths are the same,
then they can decompose into the same set of partial systems w.r.t.\ the relation $\rightsquigarrow$.

\begin{lemma}
	\label{lmm:X-Y-dec-Z}
	Given a trace-preserving distributed quantum system $S\in \Sys$,
	for any two finite sets $X,Y\subseteq S/*$ of partial systems,
	if $\biguplus_{C\in X}\omega\parens*{C}=\biguplus_{C\in Y}\omega\parens*{C}$,
	then there exists a finite set $Z\subseteq S/*$ of partial system such that
	$X\rightsquigarrow^* Z$ and $Y\rightsquigarrow^* Z$.
\end{lemma}

\begin{proof}
	We construct two finite sets $W,Z$ step by step such that finally $W=Z$.
	In the construction, we require the intermediate $W,Z$ to satisfy:
	\begin{equation}
		X\rightsquigarrow^* W\wedge Y\rightsquigarrow^* Z.
		\label{eq:X-Y-dec-Z-first}
	\end{equation}

	The construction is as follows.
	\begin{enumerate}
		\item 
			\label{stp:X-Y-dec-Z-cstr-1}
			Initially, let $W=X$ and $Z=Y$.
		\item
			\label{stp:X-Y-dec-Z-cstr-2}
			Repeat the following procedure.
			First pick $C\in W$ and $D\in Z$ such that 
			\begin{equation}
				\omega\parens*{C}\cap \omega\parens*{D}\neq \emptyset\wedge C\neq D.\label{eq:X-Y-dec-Z-choose-C-D}
			\end{equation}
			If such $C,D$ do not exist,
			then by \Cref{eq:X-Y-dec-Z-first} and \Cref{lmm:set-partial-sys-dec},
			it is easy to see that $W=Z$ and $Z$ satisfies the properties in \Cref{lmm:X-Y-dec-Z},
			and we can terminate.

			Otherwise, by \Cref{lmm:C-D-dec-X-Y},
			there exist $X_1,X_2\subseteq S/*$ such that
			\begin{enumerate}
				\item 
					$\braces*{C}\rightsquigarrow^* X_1$ and $\braces*{D}\rightsquigarrow^* X_2$.
				\item
					$X_1\cap X_2=\braces*{E}$ for some $E\in S/*$ and $\omega\parens*{E}=\omega\parens*{C}\cap \omega\parens*{D}$.
			\end{enumerate}
			Let $W'=\parens*{W-\braces*{C}}\cup X_1$ and $Z'=\parens*{Z-\braces*{D}}\cup X_2$,
			where $W',Z'$ stand for the new values of $W,Z$, respectively.
	\end{enumerate}
	
	Now we show that \Cref{eq:X-Y-dec-Z-first} holds in the above construction.
	In Step~\ref{stp:X-Y-dec-Z-cstr-1}, it obviously holds for the initial $W,Z$.
	Let us verify \Cref{eq:X-Y-dec-Z-first} for $W',Z'$ given that it holds for $W,Z$
	in Step~\ref{stp:X-Y-dec-Z-cstr-2}.
	As $\braces*{C}\rightsquigarrow^* X_1$ and $\braces*{D}\rightsquigarrow^* X_2$,
	we have $W\rightsquigarrow^* W'$ and $Z\rightsquigarrow^* Z'$ by \Cref{def:rel-squiar}.
	The conclusion immediately follows.

	It only remains to show that the above construction terminates.
	Note that in Step~\ref{stp:X-Y-dec-Z-cstr-2}, 
	we have $\forall F\in X_1\cup X_2, \ell\parens*{F}\leq \max\braces*{\ell\parens*{C},\ell\parens*{D}}$,
	according to the proof of \Cref{lmm:C-D-dec-X-Y}.
	Let $r=\max_{F\in X\cup Y}\ell\parens*{F}$.
	By a simple induction, it can be seen that
	$\forall F\in \parens*{W'-W}\cup \parens*{Z'-Z}, \ell\parens*{F}\leq r$.
	Moreover, consider the set
	\begin{equation*}
		V=\braces*{F\in S/*:\ell\parens*{F}\leq r}.
	\end{equation*}
	Let us focus on the construction of $W$.
	Each time when $W'\neq W$ in Step~\ref{stp:X-Y-dec-Z-cstr-2},
	in $W'-W$ we will visit at least one new element in $V$.
	Since $V$ is finite, 
	the number of times when $W'\neq W$ should be finite.
	A similar statement holds for $Z$.
	Hence, the above construction terminates.
\end{proof}

\subsection{Proof of \Cref{lmm:calS-semiring}}
\label{sub:proof_of_lmm_calS_semiring}

In the following we prove \Cref{lmm:calS-semiring} by verifying every condition of semiring in \Cref{def:semi-ring}.

\begin{proof}[Proof of \Cref{lmm:calS-semiring}]
	Let us verify the following properties of $\calS$:
	\begin{enumerate}
		\item 
			$\emptyset \in \calS$.

			This is from the definition of $\calS$.
		\item
			For any $X,Y\in \calS$, 
			the difference $X-Y$ is a finite disjoint union of sets in $\calS$.

			W.l.o.g, suppose that $X\cap Y\neq \emptyset$, $X=\omega\parens*{C}$ and $Y=\omega\parens*{D}$ for some $C,D\in S/*$.
			By \Cref{lmm:C-D-dec-X-Y}, there exist two finite set $W_C,W_D\subseteq S/*$
			such that 
			\begin{enumerate}
				\item 
					$\braces*{C}\rightsquigarrow^* W_C$ and $\braces*{D}\rightsquigarrow^* W_D$;
				\item
					$W_C\cap W_D=\braces*{E}$ for some $E\in S/*$ and $\omega\parens*{E}=\omega\parens*{C}\cap \omega\parens*{D}$.
			\end{enumerate}
			From \Cref{lmm:set-partial-sys-dec}, we further have $\omega\parens*{C}=\biguplus_{F\in W_C}\omega\parens*{F}$.
			Hence,
			\begin{equation*}
				X-Y=\omega\parens*{C}-\omega\parens*{C}\cap\omega\parens*{D}=\biguplus_{F\in W_C-\braces*{E}}\omega\parens*{F},
			\end{equation*}
			which is a finite disjoint union of sets in $\calS$.
		\item
			For any $X,Y\in \calS$, $X\cap Y\in \calS$.
			
			W.l.o.g, suppose that $X\cap Y\neq \emptyset$, $X=\omega\parens*{C}$ and $Y=\omega\parens*{D}$ for some $C,D\in S/*$.
			The desired property immediately follows from \Cref{lmm:C-D-dec-X-Y}.
	\end{enumerate}
	By \Cref{def:semi-ring}, $\calS$ is a semiring.
\end{proof}

\subsection{Technical Lemmas about Semiring $\calS$}
\label{sub:technical_lemmas_about_semiring_cals_}

In this section we prove two lemmas about the semiring $\calS$,
which are useful in deriving the \mbox{$\sigma$-subadditivity} of the function $\mu_{\rho\to S}$
when we prove \Cref{lmm:mu-is-pr-meas} in \Cref{sub:proofs_of_mu_is_pr_meas}.
We starts with a metric $d$ on the set of maximal paths $\omega\parens*{A}$ of a quantum process $A$,
as follows.

\begin{definition}[Metric $d$]
	For a quantum process $A$, we define a metric $d:\omega\parens*{A}\times \omega\parens*{A}\rightarrow \R_{\geq 0}$
	such that for any $B,C\in \omega\parens*{A}$,
	$d\parens*{B,C}=2^{-\abs*{B\cap C}}$ if $B\neq C$;
	and $d\parens*{B,C}=0$ otherwise.
\end{definition}

It is easy to verify $d$ is indeed a metric.
For a quantum process $A$, $\parens*{\omega\parens*{A},d}$ forms a metric space.
Note that for any partial process $B\in A/*$,
the set $\omega\parens*{B}$ is open,
because $\omega\parens*{B}=\braces*{C\in \omega\parens*{A}:d\parens*{C,C_B}<2^{-\ell\parens*{B}-2}}$
is an open ball of radius $2^{-\ell\parens*{B}-2}$ centered at $C_B$,
for any $C_B\in \omega\parens*{B}$,
where $\ell\parens*{\cdot}$ is defined in \Cref{def:func-ell}.
Similarly, 
the set $\omega\parens*{B}$ is also closed,
because $\omega\parens*{B}=\braces*{C\in \omega\parens*{A}:d\parens*{C,C_B}\leq 2^{-\ell\parens*{B}-1}}$
is a closed ball of radius $2^{-\ell\parens*{B}-1}$ centered at $C_B$,
for any $C_B\in \omega\parens*{B}$.

The following lemma shows the compactness of the metric space $\parens*{\omega\parens*{A},d}$.

\begin{lemma}
	The metric space $\parens*{\omega\parens*{A},d}$ is compact.
\end{lemma}

\begin{proof}
	It suffices to show that $\parens*{\omega\parens*{A},d}$ is complete and totally bounded, as follows.
	\begin{enumerate}
		\item 
			$\parens*{\omega\parens*{A},d}$ is complete;
			that is, every Cauchy sequence in $\omega\parens*{A}$ converges in $\omega\parens*{A}$.

			For any $B\in \omega\parens*{A}$, let $a_k\parens*{B}\in A$ be such that $\rt\parens*{A}\rightarrow^k a_k\parens*{B}$.
			Consider a Cauchy sequence $B_1,B_2,\ldots\in \omega\parens*{A}$;
			that is, $\forall \epsilon>0, \exists K\in \N$
			such that $\forall k,l>K$, $d\parens*{B_{k},B_l}<\epsilon$.
			It is easy to verify by definition that $\forall k\in \N,\exists K\in \N, b_k\in A, \forall l>K,a_k\parens*{B_l}=b_k$.
			Let $B\in \omega\parens*{A}$ be such that $a_k\parens*{B}=b_k$ for $k\in \N$.
			Then, $\lim_{k\to\infty} B_k=B$.
		\item
			$\parens*{\omega\parens*{A},d}$ is totally bounded;
			that is, $\forall \epsilon >0$, $\omega\parens*{A}$ can be covered by finitely many open balls of radius $\epsilon$.

			For any $\epsilon>0$, let $m\in\N$ be such that $2^{-m}< \epsilon$.
			For any $a\in A$ with $\rt\parens*{A}\rightarrow^m a$,
			let us pick $B_a\in \omega\parens*{A}$ such that $a\in B_a$.
			There are finitely many such $B_a$.
			Let $C_a=\braces*{B\in \omega\parens*{A}:d\parens*{B,B_a}<\epsilon}$ be the open of radius $\epsilon$ centered at $B_a$,
			then it is easy to see that 
			\begin{equation*}
				\omega\parens*{A}\subseteq \bigcup_{a\in A:\rt\parens*{A}\rightarrow^m a}C_a.
			\end{equation*}
	\end{enumerate}
\end{proof}

Recall that for any quantum system $C=A\parallel B$, $\omega\parens*{C}=\omega\parens*{A}\times \omega\parens*{B}$.
For any trace-preserving quantum system $S$, we can consider the product topology on $\omega\parens*{S}$.
As $\omega\parens*{S}$ is a product of compact spaces,
by Tychonoff theorem, it is also compact.
In this case, for any partial system $C\in S/*$, $\omega\parens*{C}$ as a finite product of clopen sets is also clopen.
Consequently, $\omega\parens*{C}$ is also compact.

Finally, we can present the two technical lemmas in this section.

\begin{lemma}
	\label{lmm:simga-subadd-to-fin}
	For the semiring $\calS$, given $X\subseteq \bigcup_{k\in \N}X_k$ with $X,X_k\in \calS$,
	there exists $K\in \N$ such that $X\subseteq \bigcup_{k\in [K]}X_k$.
\end{lemma}

\begin{proof}
	W.l.o.g., suppose that $X=\omega\parens*{C}$ and $X_k=\omega\parens*{C_k}$ for some $C,C_k\in S/*$.
	As shown above, $\omega\parens*{C}$ as a compact set has an open cover $\bigcup_{k\in \N}\omega\parens*{C_k}$.
	By the definition of compactness, there exists a finite subcover
	and the conclusion immediately follows.
\end{proof}

\begin{lemma}
	\label{lmm:subadd-to-add}
	For the semiring $\calS$, given $X\subseteq \bigcup_{k\in [K]}X_k$ with $X,X_k\in \calS$
	and $K\in \N$, there exists finite sets $P_k\subseteq \calS$ for $k\in [K]$,
	such that $\forall k\in [K], \biguplus_{Y\in P_k}Y\subseteq X_k$ and $X=\biguplus_{k\in [K]}\biguplus_{Y\in P_k}Y$.
\end{lemma}

\begin{proof}
	For $k\in [K]$, let $Z_k=X\cap X_k-\bigcup_{j<k}X_j$.
	In this case, $X=\biguplus_{k\in [K]}Z_k$,
	where $Z_k\subseteq X_k$.
	As $\calS$ is a semiring, $Z_k$ is a finite disjoint union of sets in $\calS$,
	and the conclusion immediately follows.
\end{proof}

They above two lemmas together enable one to derive the $\sigma$-subadditivity of a function $\mu$ on $\calS$
from the additivity of $\mu$.

\subsection{Proof of \Cref{lmm:mu-is-pr-meas}}
\label{sub:proofs_of_mu_is_pr_meas}

In this section we prove \Cref{lmm:mu-is-pr-meas},
using the conditions of the system dynamics in \Cref{def:sys-dyn}
and lemmas proved above.

\begin{proof}[Proof of \Cref{lmm:mu-is-pr-meas}]
	Given an initial state $\rho$, define a function $\mu:\calS\rightarrow \bracks*{0,1}$ such that
	$\mu\parens*{\emptyset}=0$, and for $C\in S/*$,
	\begin{equation*}
		\mu\circ \omega\parens*{C}=\lim_{t\to\infty} \tr\parens*{\Bracks*{C}\parens*{t}\parens*{\rho}}.
	\end{equation*}
	The limit exists because 
	$\tr\parens*{\Bracks*{C}\parens*{t}\parens*{\rho}}$
	is non-increasing with respect to $t$, according to \Cref{den:sys-dyn-tr};
	and bounded below by $0$, because $\Bracks*{C}\parens*{t}\in \QO\parens*{\calH_{C}}$ according to \Cref{def:sys-dyn}.
	In the following we show that $\mu$ satisfies the conditions for \Cref{lmm:cara ext},
	and therefore can be uniquely extended to the desired probability measure $\mu_{\rho\to S}$.
	To this end, we only need to verify the following properties:
	\begin{enumerate}
		\item 
			$\mu\parens*{\emptyset}=0$ and $\mu\circ \omega\parens*{S}=1$.

			From the definition of $\mu$, it is obvious that $\mu\parens*{\emptyset}=0$.
			From \Cref{den:sys-dyn-init-cnd,den:sys-dyn-tr},
			it is also easy to obtain $\mu\circ \omega\parens*{S}=\tr\parens*{\Bracks*{S}\parens*{0}\parens*{\rho}}=\tr\parens*{\rho}=1$.
		\item
			\label{stp:mu-is-meas-prf-squiar-cons}
			For any two finite sets $P,Q\subseteq S/*$ with $P\rightsquigarrow^* Q$,
			we have $\sum_{C\in P}\mu\circ\omega\parens*{C}=\sum_{C\in Q}\mu\circ\omega\parens*{C}$.

			By a simple induction on the transitive closure $\rightsquigarrow^*$,
			it suffices to prove this property with $\rightsquigarrow^*$ replaced by $\rightsquigarrow$.
			In this case, since $P\rightsquigarrow Q$, by \Cref{def:rel-squiar},
			we have $P-Q=\braces*{A/a\parallel B}$ for some trace-preserving $A\in \Proc$,
			$a\in A$ and $B\in \Sys$;
			and $Q-P=\braces*{A/b\parallel B:a\rightarrow b}$.
			Now it suffices to show that 
			\begin{equation}
				\mu\circ\omega\parens*{A/a\parallel B}=\sum_{b:a\rightarrow b}\mu\circ\omega\parens*{A/b\parallel B}.
				\label{eq:mu-br-sum}
			\end{equation}
			Let us choose any sufficiently large
			$t$ such that $t\geq \max_{b:a\rightarrow b}\max T\bracks*{b}$ 
			and $A/b\parallel B$ is trace-preserving after time $t$.
			By \Cref{den:sys-dyn-tr} and the definition of $\mu$, we have 
			\begin{align*}
				\mu\circ\omega\parens*{A/a\parallel B}&=\tr\parens*{\Bracks*{A/a\parallel B}\parens*{t}\parens*{\rho}},\\
				\mu\circ\omega\parens*{A/b\parallel B}&=\tr\parens*{\Bracks*{A/b\parallel B}\parens*{t}\parens*{\rho}},
			\end{align*}
			for all $b$ with $a\rightarrow b$.
			The above together with \Cref{den:sys-dyn-br} yield \Cref{eq:mu-br-sum}.
		\item
			\label{stp:mu-is-meas-prf-additive}
			$\mu$ is additive; that is,
			$\mu\parens*{X}=\sum_{k=1}^K \mu\parens*{X_k}$
			for any $X\in \calS$ and finitely many $X_1, X_2,\ldots,X_K\in \calS$ with $X=\biguplus_{k=1}^K X_k$ and $K\in \N$.

			W.l.o.g., suppose that
			$X=\omega\parens*{C}$ and $X_k=\omega\parens*{C_k}$ for some $C,C_k\in S/*$ and $k\in [K]$.
			Let $P=\braces*{C_k}_{k\in [K]}\subseteq S/*$.
			Since $\omega\parens*{C}=\biguplus_{D\in P}\omega\parens*{D}$,
			using \Cref{lmm:X-Y-dec-Z}, there exists a finite set $Q\subseteq S/*$ such that
			$\braces*{C}\rightsquigarrow^* Q$ and $P\rightsquigarrow^* Q$.
			By Property~\ref{stp:mu-is-meas-prf-squiar-cons} just proved above,
			we have $\mu\circ\omega\parens*{C}=\sum_{D\in Q}\mu\circ\omega\parens*{D}$
			and $\sum_{D\in P}\mu\circ\omega\parens*{D}=\sum_{D\in Q}\mu\circ\omega\parens*{D}$.
			Consequently, $\mu\parens*{X}=\mu\circ\omega\parens*{C}=\sum_{k\in [K]}\mu\circ\omega\parens*{C_k}=\sum_{k\in [K]}\mu\parens*{X_k}$.

		\item
			$\mu$ is $\sigma$-subadditive;
			that is,
			$\mu\parens*{X}\leq \sum_{k\in \N} \mu\parens*{X_k}$
			for any $X\in \calS$ and countably many $X_1, X_2,\ldots\in \calS$ with $X\subseteq \bigcup_{k\in \N} X_k$.

			Using \Cref{lmm:simga-subadd-to-fin,lmm:subadd-to-add},
			there exist $P_1,P_2,\ldots, P_K\subseteq \calS$ for some $K\in \N$ such that
			$\forall k\in [K], \biguplus_{Y\in P_k}Y\subseteq X_k$ and $X=\biguplus_{k\in [K]}\biguplus_{Y\in P_k} Y$.
			Since $\calS$ is a semiring, for any $k\in [K]$, there exists a finite set $Q_k\subseteq \calS$
			such that $X_k=\biguplus_{Y\in P_k} Y\uplus \biguplus_{Z\in Q_k}Z$.
			By Properties~\ref{stp:mu-is-meas-prf-additive} just proved above,
			this implies $\sum_{Y\in P_k}\mu\parens*{Y}\leq \sum_{Y\in P_k}\mu\parens*{Y}+\sum_{Z\in Q_k}\mu\parens*{Z}=\mu\parens*{X_k}$.
			Finally, we have 
			\begin{equation*}
				\mu\parens*{X}=\sum_{k\in [K]}\sum_{Y\in P_k}\mu\parens*{Y}\leq \sum_{k\in [K]}\mu\parens*{X_k}\leq \sum_{k\in\N}\mu\parens*{X_k}.
			\end{equation*}
	\end{enumerate}
\end{proof}

\subsection{Proofs of \Cref{thm:local-ins,thm:local-atom}}
\label{sub:proof_of_thm_local_ins_atom}

In this section, we are finally ready to prove the main \Cref{thm:local-ins,thm:local-atom}.
Before proving \Cref{thm:local-ins}, we will need the following lemma,
which shows that for any $t\in \R_{\geq 0}$,
a partial system can be decomposed w.r.t.\ the relation $\rightsquigarrow$,
such that after the decomposition, any partial system in the set has no branching
when restricted to the time region $[0,t]$.

\begin{lemma}
	\label{lmm:C-exp-to-after-t}
	For any trace-preserving distributed quantum system $S\in \Sys$,
	partial system $C\in S/*$, and $t\in \R_{\geq 0}$,
	there exists a finite set $P\subseteq S/*$ such that
	$\braces*{C}\rightsquigarrow^* P$ and for any $D\in P$,
	$D\restriction_{[0,t]}$ has no branching.
\end{lemma}

\begin{proof}
	We show how to construct $P$ step by step.
	In the construction, we require the intermediate $P$ to satisfy 
	\begin{equation}
		\braces*{C}\rightsquigarrow^* P.
		\label{eq:C-exp-to-after-t}
	\end{equation}

	The construction is as follows.
	\begin{enumerate}
		\item 
			Initially, let $P=\braces*{C}$.
		\item
			\label{stp:C-exp-to-after-t-2}
			Repeat the following procedure.
			First pick $D\in P$ such that $D\restriction_{[0,t]}$ has branching.
			If such $D$ does not exist, then we can terminate.
			Otherwise, we have $D=A/a\parallel B$ for some $A\in \Proc$, $a\in A$ and $B\in \Sys$
			with $\parens*{A/a}\restriction_{[0,t]}$ has branching.
			As a result, $\braces*{b\in A:a\rightarrow b}\neq \emptyset$.
			Let $P'=\parens*{P-D}\cup \braces*{A/b\parallel B:a\rightarrow b}$,
			where $P'$ stands for the new value of $P$.
	\end{enumerate}
	
	It is easy to see that \Cref{eq:C-exp-to-after-t} holds in the above construction.
	It remains to show that the above construction terminates.
	For any $E\in S/*$, let
	\begin{equation*}
		\phi\parens*{E}=\#\braces*{(c,d):c,d\in E\restriction_{[0,t]}\wedge c\not\rightarrow^* d\wedge d\not\rightarrow^* c}
	\end{equation*}
	be the number of action pair $(c,d)$ that cannot be ordered by $\rightarrow^*$ in $E\restriction_{[0,t]}$.
	Note that in Step~\ref{stp:C-exp-to-after-t-2},
	$\sum_{b:a\rightarrow b} \phi\parens*{A/b\parallel B}<\phi\parens*{A/a\parallel B}$.
	Consequently, $\sum_{E\in P'}\phi\parens*{E}<\sum_{E\in P}\phi\parens*{E}$.
	Initially, $\phi\parens*{C}$ is finite due to \Cref{def:quantum-proc}.
	Hence, the above construction terminates.
\end{proof}


Now we prove the first main theorem.

\begin{proof}[Proof of \Cref{thm:local-ins}]
	The proof consists of two steps. 
	\begin{enumerate}
		\item 
			\label{stp:local-ins-prf-1}
			Let us fix a local action $a\in S$.
			Suppose that $T\bracks*{a}=[x,y]$,
			and $x>0$ w.l.o.g.
			Consider another trace-preserving system $S'$ and an isomorphism $\gamma:S\rightarrow S'$ such that:
			\begin{enumerate}
				\item
					$e\bracks*{\gamma\parens*{a}}=\emptyset$,
					and $T\bracks*{\gamma\parens*{a}}=\braces*{t_a}$ for some $t_a\in T\bracks*{a}$.
				\item
					$\forall b\neq a\in S, \gamma\parens*{b}=b$.
			\end{enumerate}
			Let us prove $S\simeq S'$.

			By \Cref{def:equi-system,lmm:mu-is-pr-meas},
			it suffices to show
			for any state $\rho\in \calD\parens*{\calH_{q\bracks*{S}}}$
			and any partial system $C\in S/*$,
			\begin{equation}
				\mu_{\rho\to S} \circ\omega\parens*{C}=\mu_{\rho\to S'}\circ\omega\circ\gamma\parens*{C}.
				\label{eq:atom-1-tar}
			\end{equation}

			Let us fix a state $\rho$ and a partial system $C$.
			By \Cref{lmm:C-exp-to-after-t},
			there exists a finite $P\subseteq S/*$ such that
			$\braces*{C}\rightsquigarrow^* P$ and for any $D\in P$,
			$D\restriction_{[0,y]}$ has no branching.
			From the proof of \Cref{lmm:C-exp-to-after-t},
			it is easy to see $\braces*{\gamma\parens*{C}}\rightsquigarrow^* \gamma\parens*{P}$
			and for any $D\in P$, $\gamma\parens*{D}\restriction_{[0,y]}$ has no branching.
			Using \Cref{lmm:set-partial-sys-dec} and that
			$\mu_{\rho\to S}$ and $\mu_{\rho\to S'}$ are probability measures,
			we have
			\begin{align*}
				\mu_{\rho\to S}\circ\omega\parens*{C}&=\sum_{D\in P} \mu_{\rho\to S}\circ\omega\parens*{D}\\
				\mu_{\rho\to S'}\circ\omega\circ\gamma\parens*{C}&=\sum_{D\in P} \mu_{\rho\to S'}\circ\omega\circ\gamma\parens*{D}.
			\end{align*}
			Now proving \Cref{eq:atom-1-tar} reduces to proving
			\begin{equation}
				\mu_{\rho\to S} \circ\omega\parens*{D}=\mu_{\rho\to S'}\circ\omega\circ\gamma\parens*{D}
				\label{eq:atom-1-tar-red}
			\end{equation}
			for $D\in S/*$ with $D\restriction_{[0,y]}$ and $\gamma\parens*{D}\restriction_{[0,y]}$ having no branching.

			If $a\notin D$, then \Cref{eq:atom-1-tar-red} is trivial because $D=\gamma\parens*{D}$.
			In the following we only need to consider $a\in D$.
			In particular, assume $D=A\parallel B$ for some $A\in \Proc$ and $B\in \Sys$
			with $a\in A$. 

			Note that
			$D\restriction_I$ and $\gamma\parens*{D}\restriction_I$ has no branching for any interval $I\subseteq [0,y]$.
			Also,
			$D\restriction_I=\gamma\parens*{D}\restriction_I$ for any interval $I\subseteq [0,x)\cup (y,+\infty)$,
			and $B\restriction_I=\gamma\parens*{B}\restriction_I$ for any interval $I\subseteq [x,y]$.
			In the following, let us check $\Bracks*{D}\parens*{t}=\Bracks*{\gamma\parens*{D}}\parens*{t}$
			for $t=0,x-$ and $y$ step by step,
			each building upon the previous.
			Some condition checks are obvious and thus omitted for readability.
			\begin{enumerate}
				\item 
					The $t=0$ case is obvious by \Cref{den:sys-dyn-init-cnd}.
				\item
					The $t=x-$ case is also simple by taking $I=(0,x)$
					in the first part of \Cref{den:sys-dyn-evo} (see also \Cref{rmk:sys-dyn-evo-open-int}),
					combined with $D\restriction_I=\gamma\parens*{D}\restriction_I$.
				\item
					The $t=y$ case is more complicated.
					Let $I_1=[x,t_a),I_2=[t_a,t_a]=\braces*{t_a}$ and $I_3=(t_a,y]$.
					Note that $A\restriction_{[x,y]}=\braces*{a}$ and $a$ is local.
					By taking $I=[x,y]$ in the second part of \Cref{den:sys-dyn-evo},
					we have $\Bracks*{D}\parens*{y}=\parens*{\calE\bracks*{a}\otimes \calF'}\circ\Bracks*{\gamma\parens*{D}}\parens*{x-}$,
					for some $\calF'$ uniquely determined by $B\restriction_{[x,y]}$.
					Alternatively, by taking $I=I_1,I_2,I_3$ in the second part of \Cref{den:sys-dyn-evo},
					we have
					\begin{align*}
						\Bracks*{D}\parens*{t_a-}&=\parens*{\calF_1\otimes \calF_1'}\circ\Bracks*{D}\parens*{x-}\\
						\Bracks*{D}\parens*{t_a}&=\parens*{\calF_2\otimes \calF_2'}\circ\Bracks*{D}\parens*{t_a-}\\
						\Bracks*{D}\parens*{y}&=\parens*{\calF_3\otimes \calF_3'}\circ\Bracks*{D}\parens*{t_a},
					\end{align*}
					for some quantum operations $\calF_b$ and $\calF_b'$ uniquely determined 
					by $A\restriction_{I_b}$ and $B\restriction_{I_b}$ with $b\in [3]$,
					respectively.
					As a result, $\calF'=\calF_3'\circ\calF_2'\circ\calF_1'$.

					Note that $\gamma\parens*{A}\restriction_{[x,t_a)}=\gamma\parens*{A}\restriction_{(t_a,y]}=\emptyset$ since $a$ is local.
					By taking $I=I_1, I_3$ in the second part of \Cref{den:sys-dyn-evo},
					combined with $B\restriction_I=\gamma\parens*{B}\restriction_I$,
					we have 
					\begin{align*}
						\Bracks*{\gamma\parens*{D}}\parens*{t_a-}&=\parens*{\Id\otimes\calF_1'}\circ\Bracks*{\gamma\parens*{D}}\parens*{x-}\\
						\Bracks*{\gamma\parens*{D}}\parens*{y}&=\parens*{\Id\otimes\calF_3'}\circ\Bracks*{\gamma\parens*{D}}\parens*{t_a}.
					\end{align*}
					By taking $I=I_2$ in the second part of \Cref{den:sys-dyn-evo},
					combined with $B\restriction_I=\gamma\parens*{B}\restriction_I$,
					we have $\Bracks*{\gamma\parens*{D}}\parens*{t_a}=\parens*{\calE\bracks*{\gamma\parens*{a}}\otimes \calF_2'}\circ\Bracks*{\gamma\parens*{D}}\parens*{t_a-}$.
					Since $\gamma$ is an isomorphism, $\calE\bracks*{\gamma\parens*{a}}=\calE\bracks*{a}$.
					
					The above together yield the conclusion.
			\end{enumerate}
			For any $t>y$, we also have $\Bracks*{D}\parens*{t}=\Bracks*{\gamma\parens*{D}}\parens*{t}$,
			by taking $I=(y,t]$ in the second part of \Cref{den:sys-dyn-evo},
			combined with $D\restriction_I=\gamma\parens*{D}\restriction_I$ and the above results.

			Finally, from \Cref{lmm:mu-is-pr-meas}, we have
			\begin{align*}
				\mu_{\rho\to S}\circ\omega\parens*{D}&=\lim_{t\to\infty} \tr\parens*{\Bracks*{D}\parens*{t}\parens*{\rho}}\\
				\mu_{\rho\to S'}\circ\omega\circ\gamma\parens*{D}&=\lim_{t\to\infty}\tr\parens*{\Bracks*{\gamma\parens*{D}}\parens*{t}\parens*{\rho}},
			\end{align*}
			and \Cref{eq:atom-1-tar-red} immediately follows.
		\item
			Now we construct a system $S'$ and an isomorphism $\gamma:S\rightarrow S'$
			that satisfy the properties in \Cref{thm:local-ins}.
			Since $S$ as a set is countable,
			we can enumerate all local operations in $S$ as $\braces*{a_1,a_2,\ldots}$.
			Consider a set of systems $\braces*{S_m}_{m\in \N}$ with each $S_m\in \Sys$
			and a set of isomorphism $\braces*{\eta_m}_{m\in \N}$ with each $\eta_m:S_m\rightarrow S_{m+1}$ such that
			$S_1=S$ and for any $m>1$:
			\begin{enumerate}
				\item 
					$e\bracks*{\eta_m\parens*{a_m}}=\emptyset$ 
					and $T\bracks*{\eta_m\parens*{a_m}}=\braces*{t_{m}}$
					for some $t_{m}\in T\bracks*{a_m}$.
				\item
					$\forall b\neq a_m\in S_m, \eta_m\parens*{b}=b$.
			\end{enumerate}
			According to the results in Step~\ref{stp:local-ins-prf-1},
			we have $\forall m\in \N, S_m\simeq S$.
			Let $\gamma_m=\eta_m\circ\eta_{m-1}\circ\ldots\circ\eta_1$,
			then $\gamma_m:S\rightarrow S_m$ is an isomorphism.
			Let $\gamma=\lim_{m\to\infty} \gamma_m$ be the point-wise limit of $\gamma_m$.
			It can be seen that $\gamma$ is an isomorphism and together with the system $S'=\gamma\parens*{S}$
			satisfy the first and second properties in \Cref{thm:local-ins}.

			It remains to show that $S'\simeq S$;
			that is, by \Cref{def:equi-system},
			to show for any state $\rho$ and partial system $C\in S/*$,
			\begin{equation}
				\mu_{\rho\to S}\circ \omega \parens*{C}
				=\mu_{\rho\to S'}\circ\omega\circ\gamma\parens*{C}.
				\label{eq:atom-2-tar}
			\end{equation}
			As $\forall m\in \N,S_m\simeq S$,
			proving \Cref{eq:atom-2-tar} reduces to proving the existence of some $m'\in \N$ such that
			\begin{equation}
				\mu_{\rho\to S_{m'}}\circ\omega\circ\gamma_{m'}\parens*{C}
				=\mu_{\rho\to S'}\circ\omega\circ\gamma\parens*{C}.
				\label{eq:atom-2-tar-red}
			\end{equation}
			To this end,
			let us first choose a sufficiently large $t\in \R_{\geq 0}$ 
			such that $C$ is trace-preserving after time $t$.
			Since for any $a\in C$ and $m\in\N$,
			$T\bracks*{\gamma_m\parens*{a}}\subseteq T\bracks*{a}$,
			we have that $\gamma_m\parens*{C}$ and $\gamma\parens*{C}$ are also trace-preserving after time $t$.
			According to \Cref{lmm:mu-is-pr-meas} and \Cref{den:sys-dyn-tr}, proving \Cref{eq:atom-2-tar-red} further reduces to proving
			\begin{equation}
				\tr\parens*{\Bracks*{\gamma_{m'}\parens*{C}}\parens*{t}\parens*{\rho}}
				=\tr\parens*{\Bracks*{\gamma\parens*{C}}\parens*{t}\parens*{\rho}}.
				\label{eq:atom-2-tar-fur-red}
			\end{equation}
			Let us choose a sufficiently large $m'\in \N$ such that 
			$\gamma_{m'}\parens*{C}\restriction_{(0,t]}=\gamma\parens*{C}\restriction_{(0,t]}$,
			which is achievable because
			the set $C\restriction_{(0,t]}$ is finite and 
			$\gamma_m\parens*{C}\restriction_{(0,t]}\subseteq C\restriction_{(0,t]}$ for any $m\in \N$.
			Now by taking $I=(0,t]$ in \Cref{den:sys-dyn-evo},
			we have $\Bracks*{\gamma_{m'}\parens*{C}}\parens*{t}=\calF\circ\Bracks*{\gamma_{m'}\parens*{C}}\parens*{0}$
			and $\Bracks*{\gamma\parens*{C}}\parens*{t}=\calF\circ\Bracks*{\gamma\parens*{C}}\parens*{0}$
			for some quantum operation $\calF$ uniquely determined by 
			$\gamma_{m'}\parens*{C}\restriction_{(0,t]}=\gamma\parens*{C}\restriction_{(0,t]}$.
			Combined with $\Bracks*{\gamma_{m'}\parens*{C}}\parens*{0}=\Bracks*{\gamma\parens*{C}}\parens*{0}$ from \Cref{den:sys-dyn-init-cnd},
			we have $\Bracks*{\gamma_{m'}\parens*{C}}\parens*{t}=\Bracks*{\gamma\parens*{C}}\parens*{t}$,
			and \Cref{eq:atom-2-tar-fur-red} immediately follows.
	\end{enumerate}
\end{proof}

Finally, we prove the second main theorem.

\begin{proof}[Proof of \Cref{thm:local-atom}]
	Let us construct a system $S'$ that satisfies the properties in \Cref{thm:local-atom}.
	Since $S$ as a set is countable,
	we can enumerate all local actions in $S$ as $L=\braces*{a_1,a_2,\ldots}$.
	Consider a system $S'$ and an isomorphism $\gamma:S\rightarrow S'$ such that
	\begin{enumerate}
		\item 
			$\forall m\in \N, e\bracks*{\gamma\parens*{a_m}}=\emptyset$ and $T\bracks*{\gamma\parens*{a_m}}=\braces*{t_m}$,
			where $t_m$ are chosen by
			\begin{equation*}
				t_m\in T\bracks*{a_m}-\braces*{t_n:n<m, a_n\in A, a_m\in B, S=A\parallel B}.
			\end{equation*}
			The existence of $t_m$ is guaranteed by $\abs*{T\bracks*{a_m}}>0$.
		\item
			$\forall b\notin L, \gamma\parens*{b}=b$.
	\end{enumerate}
	Moreover, as $S$ is aligned,
	when choosing the $\braces*{t_m}_{m\in \N}$ we can also require that
	for any $b\in S$ and $j,k\in \N$, if $b\rightarrow a_j$ and $b\rightarrow a_k$
	then $t_j=t_k$,
	which is achievable because $\min T\bracks*{a_j}=\min T\bracks*{a_k}$
	and $\abs*{T\bracks*{a_j}},\abs*{T\bracks*{a_k}}>0$.
	It is easy to see that $S'$ satisfies the properties in \Cref{thm:local-atom}.
	By the proof of \Cref{thm:local-ins},
	we also have $S'\simeq S$.
\end{proof}

\end{document}